%% file: main.tex
\documentclass{article}
\usepackage{graphicx} 
\input{preamble}

\title{\bfseries Equitable \textcolor{red}{C}\textcolor{orange}{o}\textcolor{violet}{l}\textcolor{green}{o}\textcolor{blue}{r}\textcolor{cyan}{i}\textcolor{magenta}{n}\textcolor{purple}{g}\textcolor{brown}{s} of Vertex-Weighted Graphs}
\author{Siddharth Barman\thanks{Indian Institute of Science.{\tt barman@iisc.ac.in}} \qquad \quad Vignesh Viswanathan\thanks{University of Massachusetts Amherst.{\tt vviswanathan@umass.edu}}}
\date{}

\begin{document}

\maketitle

\begin{abstract}
We study a generalization of the classical Hajnal-Szemer\'edi theorem to vertex-weighted graphs. Given a graph with nonnegative vertex weights, a coloring is called $\alpha$-approximately equitable up to one vertex ($\alpha$-\EQo) if, for each color class, the total weight remaining after removing its maximum-weight vertex is at most $\alpha \geq 1$ times the weight of any other color class. 

For vertex-weighted graphs with maximum degree $\Delta$, we show that there exist instances for which no $k$-coloring is $\alpha$-\EQo for any $k < \frac{3\Delta}{2}$ and $\alpha < \sqrt{2}$. In light of this impossibility, we relax these parameters and establish the following results for any vertex-weighted graph $G$ with maximum degree $\Delta$:
\begin{enumerate}
\item For any $\vare \in (0,1)$ and all $k \geq \left( \frac{c}{\varepsilon^2}\ln{\frac{1}{\varepsilon}} \right) \Delta$, there exists a $(1 + \vare)$-\EQo $k$-coloring of $G$, where $c$ is a fixed constant; and 
\item For all $k \ge \Delta + 1$, there exists a $2$-\EQo $k$-coloring of $G$.
\end{enumerate}
Furthermore, such equitable colorings can be computed in polynomial time. En route to our results on equitability under vertex weights, we establish sufficient conditions for the existence of $k$-colorings that are equitable with respect to any given partition of the vertex set. 

Our coloring results correspond to fairness guarantees in a constrained fair division setting and lead to concentration inequalities for partly dependent random variables. 
\end{abstract}
\thispagestyle{empty}
\newpage

\tableofcontents
\thispagestyle{empty}
\newpage

\section{Introduction}
Resolving a conjecture by Erd\H{o}s, Hajnal and Szemer\'edi \cite{HajnalSzemeredi1970} proved that every graph with maximum degree $\Delta$ admits an {\em equitable} $(\Delta+1)$-coloring, i.e., a coloring in which the sizes of all color classes differ by at most one. The current work develops extensions of this result to the weighted setting.

Given a graph $G = (V, E, w)$ with nonnegative vertex weights $\{w_v\}_{v \in V}$, a $k$-coloring $(C_1, \dots, C_{k})$ is said to be $\alpha$-approximately equitable up to one vertex ($\alpha$-\EQo) if, for every pair of color classes $C_i$ and $C_j$ (with $C_i \neq \emptyset$), there exists a vertex $v \in C_i$ such that $w(C_i \setminus \{v\}) \le \alpha w(C_j)$, where $w(S)=\sum_{v\in S} w(v)$ denotes the total weight of the vertices in $S \subseteq V$. A coloring is called \EQo if this condition holds with $\alpha = 1$.

We prove that, in the weighted setting, for all $k \ge \Delta + 1$ (as in \cite{HajnalSzemeredi1970}), there exists a $k$-coloring in which the color-class weights are within a factor of two, up to one vertex. Formally,

\begin{restatable}{theorem}{thmdeltaplusone}\label{thm:delta-plus-one}
Let $G = (V, E, w)$ be a vertex-weighted graph with maximum degree $\Delta$. For any $k \ge \Delta + 1$, there exists a $2$-\EQo $k$-coloring of $G$.
\end{restatable}
Ideally, one would hope to extend the Hajnal-Szemer\'{e}di theorem to vertex-weighted graphs by guaranteeing the existence of (exact) \EQo $k$-colorings for all $k \ge \Delta + 1$, i.e., to obtain  \Cref{thm:delta-plus-one} without any approximation factor. However, Theorem~\ref{thm:lower-bound} shows that such a strong extension is not possible. Specifically, this lower bound asserts that, when $k < \tfrac{3\Delta}{2}$, no $k$-coloring can achieve a guarantee better than $\sqrt{2}$-\EQo.

\begin{restatable}{theorem}{thmlowerbound}\label{thm:lower-bound}
For every odd integer $\Delta \ge 3$, there exists a vertex-weighted graph $G=(V,E,w)$ with maximum degree $\Delta$ such that, for any $k < \tfrac{3\Delta}{2}$ and $\alpha < \sqrt{2}$, no $\alpha$-\EQo $k$-coloring of $G$ exists.
\end{restatable}
In view of this impossibility, we allow the number of color classes to exceed $\Delta + 1$ and obtain results for both \EQo and $(1+\varepsilon)$-\EQo colorings in this regime. We prove that an \EQo coloring can be achieved using $O(d\Delta)$ colors, where $d$ denotes the number of distinct weights in $G$.

\begin{restatable}{theorem}{thmdequitable}\label{thm:d-equitable}
Let $G = (V, E, w)$ be a vertex-weighted graph with maximum degree $\Delta$, and let $d$ denote the number of distinct values in $\{w_v\}_{v \in V}$. For any integer $k \ge (8d+10) \Delta$, there exists an   \EQo $k$-coloring of $G$. Furthermore, such a coloring can be computed in polynomial time.
\end{restatable}

\paragraph{Main Result.} Our main technical result is the following theorem, which establishes that for any $\varepsilon > 0$, one can obtain a $(1+\varepsilon)$-\EQo coloring with $k = \widetilde{O} \left(\Delta/{\varepsilon^2}\right)$ colors.

\begin{restatable}{theorem}{thmepsequitable}\label{thm:eps-equitable}
Let $G = (V, E, w)$ be a vertex-weighted graph with maximum degree $\Delta$. For any $\vare \in (0, 0.1]$ and integer $k \ge \frac{4(1 + \varepsilon)^2}{\varepsilon^2}\left (\ln{\frac1{\varepsilon}} + 4 \right ) \Delta$, there exists a $(1 + 3\vare)$-\EQo $k$-coloring of $G$. Additionally, such a coloring can be computed in polynomial time.
\end{restatable}

\subsection{Applications}

\paragraph{Fair Division.}
Our coloring results directly imply fairness guarantees in a setting of allocating indivisible items with conflicts \cite{Chiarelli2020ConflictingItems,Hummel2022ConflictingItems}. In this discrete fair division setup, a set of $n$ indivisible items must be partitioned among $k$ agents. Conflicts among items are represented by a graph $G$ on the item set, where an edge between two items indicates a conflict. No agent can receive a pair of conflicting items, i.e., each agent must be assigned an independent set of $G$. 

Our results imply the existence of equitable allocations in this constrained setting under identical additive valuations. Equitability is a well-studied fairness notion  \cite{Dubins1961EquitableDivisible,Cechlarova2012EquitableDivisible,Freeman2019Equitable,Barman2026Equitable,Gourves2014Equitable}. In the context of indivisible items, two standard relaxations of this notion are equitability up to one item (\EQo) and equitability up to any item (EQX); see, e.g., \cite{Freeman2019Equitable,Barman2026Equitable}.

When items are chores (have costs instead of utilities) and all agents have identical costs, $\alpha$-\EQo $k$-colorings of $G$---where weights represent chore costs---correspond to $\alpha$-\EQo allocations. Consequently, our results provide sufficient conditions on the number of agents $k$ for the existence of such equitable allocations in the presence of conflicts.

One can similarly define the notion of an $\alpha$-EQX $k$-coloring, wherein the weights of the color classes must be within a factor of $\alpha$ after the removal of {\it any} vertex. While EQX allocations are known to exist in the unconstrained setting \cite{Gourves2014Equitable}, even under monotone valuations \cite{Barman2026Equitable}, we note that there exist vertex-weighted graphs for which no non-trivial $\alpha$-EQX $k$-colorings exist; see Appendix \ref{apdx:eqx}.

\paragraph{Concentration Inequalities.} The existence of equitable colorings in the unweighted case (specifically, the Hajnal-Szemer\'edi theorem) has been used to establish tail bounds for sums of partly dependent identical Bernoulli random variables \cite{Pemmaraju2001TailBounds}. We show that, along similar lines, our weighted equitability results lead to tail bounds for non-identically distributed random variables. Specifically, we establish the following concentration inequality for the sums of partly dependent (non-identical) random variables.
\begin{restatable}{theorem}{thmtailbounds}\label{thm:tail-bounds}
Let $X_1, \dots, X_n$ be random variables such that each $X_i \in [a_i, b_i]$ almost surely, and let $G=([n],E)$ denote their dependency graph, where $(i,j) \in E$ if and only if $X_i$ and $X_j$ are dependent. Let $X=\sum_{i=1}^n X_i$, and let $\Delta$ denote the maximum degree of $G$. Then, for all $t>0$,
\begin{align*}
    \Pr \left\{ X - \E[X] \ge t \right\} \le (\Delta + 1) \exp \left (\frac{-2t^2}{2(\Delta +1)\sum_{i = 1}^n (b_i - a_i)^2 + (\Delta + 1)^2 \max_{i \in [n]} (b_i - a_i)^2} \right ).
\end{align*}
The same upper bound holds for $\Pr\left\{X - \E[X] \le -t \right\}$.
\end{restatable}
We note that these bounds are not novel and follow from the deviation inequalities in \cite{Janson2004TailBounds}. However, the proof technique in \cite{Janson2004TailBounds} is different and, in particular, does not rely on equitable colorings. Although the tail bounds here are slightly weaker, we believe that our alternative proof may be of independent interest.

\subsection{Additional Results}
We also obtain improved guarantees when the maximum vertex weight is small relative to the total weight of the vertex set. Specifically, the following result applies when $\max_{v \in V} w_v \le \frac{\varepsilon}{k} \sum_{v \in V} w_v$ for some small $\varepsilon > 0$.

\begin{restatable}{theorem}{thmchromaticnumber}\label{thm:chromatic-number}
Let $G = (V, E, w)$ be a vertex-weighted graph with maximum degree $\Delta$. For any parameter $\vare \in (0, 0.1]$ and integer $k \in \mathbb{Z}_+$, if $k \ge \max \left \{\frac{\chi(G)}{\vare}, 2\Delta \right \}$ and $\max_{v \in V} w_v \le \vare \frac{w(V)}{k}$, then there exists a $(1 + 7\vare)$-\EQo $k$-coloring of $G$, where $\chi(G)$ denotes the chromatic number of $G$. 
\end{restatable}
Assuming low maximum weight, this result strengthens Theorem~\ref{thm:eps-equitable} in terms of the dependence on $\varepsilon$ and $\Delta$. However, in contrast to Theorem~\ref{thm:eps-equitable}, it is not accompanied by a polynomial-time algorithm. We also show that, by weakening the condition on $k$, one can obtain a comparable result with a polynomial-time algorithm.    
\begin{restatable}{theorem}{thmlowmaxwtcomputable}\label{thm:low-max-wt-computable}
Let $G = (V, E, w)$ be a vertex-weighted graph with maximum degree $\Delta$. For any parameter $\vare \in (0, 0.1]$ and integer $k \in \mathbb{Z}_+$, if $k \ge \frac{\Delta + 1}{\vare}$ and $\max_{v \in V} w_v \le \vare \frac{w(V)}{k}$, then $G$ admits a $(1 + 7\vare)$-\EQo $k$-coloring that can be computed in polynomial time.
\end{restatable}

\subsection{Overview of Proof Techniques}
\paragraph{Equitability with respect to a partition.} 
To establish equitability results in the weighted setting (Theorems~\ref{thm:d-equitable} and \ref{thm:eps-equitable}), a key idea in this work is to first prove the existence of colorings that are equitable with respect to a given partition $V_1,\dots,V_d$ of the vertices $V$. Formally, we establish the following theorem on partition equitability. 
\begin{restatable}{theorem}{thmpartitionequitability}\label{thm:partition-equitability}
Given a graph $G = (V, E)$ with maximum degree $\Delta \in \mathbb{Z}_+$, let $(V_1, V_2, \ldots, V_d)$ be any partition of the vertex set $V$ satisfying $|V_1| \le |V_2| \le \ldots \le |V_d|$. Also, let $k \in \mathbb{Z}_{+}$ be any integer such that $k \ge (4\eta + 2)\Delta$, where 
\begin{align*}
    \eta \coloneqq \max_{t \in [d]} \left( \frac{\sum_{j = 1}^{t} |V_j|}{|V_t|} \right).
\end{align*}
Then, the graph $G$ admits a $k$-coloring $(C_1,\ldots, C_k)$ with the property that $|C_i \cap V_j| \in \left \{ \left \lfloor  \frac{|V_j|}{k} \right \rfloor, \left \lceil  \frac{|V_j|}{k} \right \rceil \right \}$ for all $i \in [k]$ and $j \in [d]$. Furthermore, such a coloring can be computed in polynomial time.
\end{restatable}
This theorem can be viewed as a graph-theoretic variant of the discrete necklace-splitting problem~\cite{alon1987splitting}. In this interpretation, each part $V_i$ corresponds to beads of type $i$, and the objective is to partition the vertex set into independent sets $C_1, \ldots, C_k$ that contain the same number of beads of each type.

The proof of \Cref{thm:partition-equitability} proceeds via an inductive path-augmentation argument. At a high level, we show that if a maintained coloring $(C_1,\dots,C_k)$ is equitable with respect to the parts $V_1,\dots,V_t$ but not $V_{t+1}$, then (using that $k$ is sufficiently large) we can modify the coloring of the vertices in $V_{t+1}$ to extend equitability to this part as well. To implement this extension, we consider a directed graph $H=([k],E')$ in which there is an edge from $i$ to $j$ if there exists a vertex in $C_i \cap V_{t+1}$ that can be moved to $C_j$ while preserving independence. We show that, if $(C_1,\dots,C_k)$ is not equitable with respect to $V_{t+1}$, then there exists a path in $H$ such that augmenting along this path improves equitability with respect to $V_{t+1}$.

\paragraph{Reduction from \EQo to partition equitability.}
Theorems~\ref{thm:d-equitable} and \ref{thm:eps-equitable} use Theorem~\ref{thm:partition-equitability} to construct colorings that are (approximately) equitable with respect to vertex weights. To see why this is possible, consider the following partition: let $V_1$ consist of the $k$ highest-weight vertices in $V$, $V_2$ consist of the $k$ highest-weight vertices in $V \setminus V_1$, and so on. Any $k$-coloring that is equitable with respect to this partition is \EQo under the weights. However, for this partition, the $\eta$ parameter in Theorem~\ref{thm:partition-equitability} may be large. Theorem~\ref{thm:d-equitable} is obtained by making minor modifications to this partition to control $\eta$. In contrast, Theorem~\ref{thm:eps-equitable} follows from an intricate and novel bucketing argument.

Specifically, for Theorem~\ref{thm:eps-equitable}, we partition the vertices into buckets $B_1, B_2, \ldots$, where each $B_j$ contains vertices whose scaled weights lie in the interval $\left( \frac{1}{(1+\varepsilon)^j}, \frac{1}{(1+\varepsilon)^{j-1}} \right]$. We then adjust these buckets by removing some vertices so that the bucket sizes grow in a roughly geometric manner. This geometrically increasing sequence of bucket sizes ensures that the $\eta$ parameter in Theorem~\ref{thm:partition-equitability} is small. In addition, we show that the total weight of the removed vertices is relatively small. We construct the final coloring in two steps: first, we apply Theorem~\ref{thm:partition-equitability} to the buckets, and then properly color the remaining vertices. Since the total weight of the removed vertices is small, the second step has only a bounded impact on the weight of each color class. This overall enables us to establish a $(1 + \vare)$-\EQo guarantee. 

\paragraph{Greedy algorithm for graphs with low maximum weight.}
Theorems~\ref{thm:chromatic-number} and \ref{thm:low-max-wt-computable} address graphs with low maximum weight, i.e., $\max_{v \in V} w_v \le \varepsilon \frac{w(V)}{k}$ for small $\varepsilon > 0$. In this setting, we present a greedy algorithm that converts any $\kappa$-coloring $(C'_1,\dots,C'_\kappa)$ into a $(1+c\varepsilon)$-\EQo $k$-coloring for any $k \ge \max \ \left\{\frac{\kappa}{\varepsilon}, 2\Delta \right\}$, where $c$ is a fixed constant.

The conversion is straightforward: starting from the coloring $(C'_1,\dots,C'_\kappa)$, we further partition each color class $C'_i$ into subsets $D_{i,1},\dots,D_{i,\ell_i}$ such that each $D_{i,j}$ (except possibly $D_{i,\ell_i}$) has weight in the interval $\left [(1 -2\vare) \frac{w(V)}{k}, (1 - \vare)\frac{w(V)}{k} \right ]$. This can be achieved via a greedy algorithm, since each vertex has weight at most $\vare \frac{w(V)}{k}$. We apply this procedure to each color class $C'_i$ to obtain a new coloring $(C_1, \dots, C_{k})$ in which every color class has weight in the interval $\left [(1 -2\vare) \frac{w(V)}{k}, (1 - \vare)\frac{w(V)}{k} \right ]$. 

This new coloring may leave some vertices uncolored. We then extend it by properly coloring the remaining vertices, and show that this step does not significantly affect the weight of each color class. Overall, this yields a $(1 + c \vare)$-\EQo $k$-coloring. 

Applying this procedure with different initial colorings $(C'_1,\ldots,C'_\kappa)$ gives us the two stated results. Theorem~\ref{thm:chromatic-number} follows by applying the procedure to a $\chi(G)$-coloring, where $\chi(G)$ denotes the chromatic number of $G$. Theorem~\ref{thm:low-max-wt-computable} follows by applying the procedure to a $(\Delta+1)$-coloring.

\subsection{Additional Related Work}
To the best of our knowledge, the only other works that study graph colorings with vertex weights are \cite{Pemmaraju2008SymmetryBreaking} and \cite{Hummel2022ConflictingItems}. Using the probabilistic method and assuming the weight of each vertex lies in the interval $[0, 1]$, \cite{Pemmaraju2008SymmetryBreaking} show that there exists a $(\Delta + 1)$-coloring $(C_1, \ldots, C_{\Delta + 1})$ such that each color class $C_j$ satisfies
\begin{align*}
w(C_j) \ge \left ( 1 - \frac1e \right )\frac{w(V)}{\Delta + 1} - 5\sqrt{w(V)\log(\Delta + 1)}.
\end{align*}
This result is not directly comparable to ours, since it establishes a lower bound on the weights but does not guarantee that the color-class weights are close to one another. Moreover, in contrast to the main results of this work, the bound from \cite{Pemmaraju2008SymmetryBreaking} is meaningful only when the maximum weight is small relative to the total weight; in particular, when $\max_{v \in V} w_v \le 1  \leq \frac{w(V)}{(\Delta + 1)^2 \log(\Delta + 1)}$.

In Appendix~\ref{apdx:probabilistic-method}, we generalize the randomized coloring method of \cite{Pemmaraju2008SymmetryBreaking} to construct approximately equitable $k$-colorings for any $k \ge \frac{\Delta+1}{\varepsilon}$. We establish the following theorem.

\begin{restatable}{theorem}{thmprobabilisticmethod}\label{thm:probabilistic-method}
Let $G = (V, E, w)$ be a vertex-weighted graph with maximum-degree $\Delta$. For any $\vare \in (0, 0.01]$, if $k \ge \frac{\Delta + 1}{\vare}$ and $\max_{v \in V} \ w_v \le \frac{\vare^2 w(V)}{k^2 \ln{k}}$, then there exists a $(1 + 25\vare)$-\EQo $k$-coloring in the graph $G$. 
\end{restatable}
This result is strictly weaker than Theorem~\ref{thm:low-max-wt-computable}, highlighting a limitation of randomized coloring for obtaining approximate \EQo guarantees.

\cite{Hummel2022ConflictingItems} study fair allocation with conflicting items, and some of their fair division results provide equitable graph coloring guarantees. Specifically, they show that an \EQo $k$-coloring exists for all $k \ge \Delta + 1$ in the following two specific cases:
\begin{enumerate}[(a)]
    \item each connected component in the graph $G$ has size at most $k$, and
    \item $G$ is a disjoint union of paths. 
\end{enumerate}
They also show that an \EQo $k$-coloring may not exist when $k = \Delta +1$; note that \Cref{thm:lower-bound} provides a stronger impossibility result. The positive results in \cite{Hummel2022ConflictingItems} rely on nontrivial ideas but apply to specific classes of graphs. By contrast, our guarantees hold for general graphs.

\paragraph{Organization.} Section~\ref{sec:delta-plus-one} presents our results on equitable colorings with $\Delta + 1$ colors. Section~\ref{sec:partition-equitability} develops our results on partition equitability and establishes key technical tools used throughout the remainder of the paper. Sections~\ref{sec:exact-eqo} and \ref{section:apx-eqo} provide guarantees for exact \EQo colorings and $(1+\varepsilon)$-\EQo colorings under vertex weights. Finally, Section~\ref{sec:low-max-wt-eqo} presents our results on $(1+\varepsilon)$-\EQo colorings for graphs with low maximum vertex weight.

The proofs of the tail bounds (\Cref{thm:tail-bounds}) and the randomized coloring result (\Cref{thm:probabilistic-method}) are deferred to the appendix, along with a counterexample for equitability up to {\em any} vertex ($\alpha$-EQX).

\section{Preliminaries}\label{sec:prelims}
We use $[k]$ to denote the set $\{1, 2,\ldots, k\}$.
We will, throughout, write $G = (V, E, w)$ to denote a vertex-weighted undirected graph with $n = |V|$ vertices and maximum degree $\Delta \in \mathbb{Z}_+$; here, the vertices $v \in V$ have nonnegative weights $w_v \in \mathbb{Q}$. For any subset of vertices $S \subseteq V$, let $w(S)$ denote the total weight of the set $S$, i.e., $w(S) = \sum_{v \in S} w_v$. Also, let $\minw(S) \coloneqq \min_{v \in S} w_v$ and $\maxw(S) \coloneqq \max_{v \in S} w_v$ respectively denote the minimum and the maximum weight among the vertices in $S$.  

This work studies the problem of partitioning the vertex set $V$ into independent subsets, i.e., subsets in which no two vertices are connected by an edge. 

\begin{definition}[Coloring]
Given a graph $G = (V, E)$, a $k$-partition $(C_1, \ldots, C_k)$ of the set of vertices $V$ is said to be a {\em $k$-coloring} (or simply a coloring) of $G$ if each $C_i$ is an independent set in $G$. 
\end{definition}
Note that the chromatic number of a graph, $\chi(G)$, is the smallest $k$ for which a $k$-coloring exists.

\begin{definition}[Partial Coloring]
Given a graph $G = (V, E)$, a  {\em partial  $k$-coloring} of $G$ is a collection of $k$ pairwise-disjoint independent subsets $C'_1,\ldots, C'_k \subseteq V$. Here, the vertices in $\cup_{i=1}^k C'_i \subset V$ are said to be colored, while the remaining vertices are uncolored. 
\end{definition}

Our overarching goal is to construct $k$-colorings such that the color classes have approximately the same weight. Formally, 

\begin{definition}[Approximate Equitability up to One Vertex --- $\alpha$-\EQo]
\label{defn:eqo}
Given a vertex-weighted graph $G = (V, E, w)$ and parameter $\alpha \ge 1$, a coloring $(C_1, \ldots, C_k)$ is said to be $\alpha$-equitable up to one vertex ($\alpha$-\EQo) if, for each $i, j \in [k]$ with $C_i \neq \emptyset$, there exists a vertex $\bar{v} \in C_i$ such that $w(C _i\setminus \{\bar v\}) \le \alpha w(C_j)$. When $\alpha = 1$, we say that the coloring $(C_1, \ldots, C_k)$ is (exact) \EQo.
\end{definition}

For any graph $G = (V, E)$ and subset of vertices $S \subseteq V$, we use $G[S]$ to denote the subgraph of $G$ induced by $S$. That is, $G[S]$ is a graph with vertex set $S$, where an edge exists between $i, j \in S$ if and only if it exists in $G$.

\section{Equitable Colorings with $\Delta + 1$ Colors}\label{sec:delta-plus-one}
This section shows that every vertex-weighted graph, with maximum degree $\Delta$, admits a $2$-\EQo coloring for any number of colors $k \geq \Delta + 1$.

\begin{algorithm}[h]
\caption{Finding a $2$-\EQo Coloring}
\label{algo:2-eqo}
\begin{algorithmic}[1]
\Require Vertex-weighted graph $G=(V, E, w)$ and integer $k \in \mathbb{Z}_+$.
\Ensure A $2$-\EQo $k$-coloring $(C^*_1, \dots, C^*_k)$. 
\State Initialize set of vertices $U = V$ and coloring $(C_1, \ldots, C_k) = (\emptyset, \ldots, \emptyset)$.
\While{there exists $j \in [k]$ and an independent set $X \subseteq U$ such that $w(C_j) < w(X)$}
    \State Let $I \subseteq U$ be a minimally envied independent set and let $i \in [k]$ be such that $w(C_i) < w(I)$.
    \State Update $C_i = I$ and $U = V \setminus \left( \cup_{a=1}^k C_a \right)$.
\EndWhile
\State Initialize $C^*_i = C_i$ for each $i \in [k]$. 
\While{$V \setminus \left( \cup_{a=1}^k C^*_a \right) \neq \emptyset$}
\State Select any vertex $u \in V \setminus \left( \cup_{a=1}^k C^*_a \right)$ and for the vertex find a color class $C^*_i$ that contains none of its neighbors; since $k \geq \Delta +1$, such a class $C^*_i$ exists. 
\State Update $C^*_i \leftarrow C^*_i \cup \{ u \}$. 
\EndWhile 
\State \Return $k$-coloring $(C^*_1, \dots, C^*_k)$.
\end{algorithmic}
\end{algorithm}

\thmdeltaplusone*
\begin{proof}
We provide a constructive proof showing that the desired coloring can be computed via an adaptation of the feasible EFX algorithm from \cite{Barman2023FeasiblEFX}.

To describe the algorithm, we first introduce some notation. Under a (partial) coloring $(C_1, \ldots, C_k)$, we say that $i \in [k]$ envies a subset of vertices $S \subseteq V$ if $w(C_i) < w(S)$. A subset $S \subseteq V$ is said to be a minimally envied if there exists some $i \in [k]$ that envies $S$, but no proper subset of $S$ is envied by any $j \in [k]$.

The $2$-\EQo algorithm (Algorithm \ref{algo:2-eqo}) is as follows. We begin with the empty partial coloring $(C_1, \dots, C_k)$, where each $C_i = \emptyset$, and the set of uncolored vertices $U = V$. In each iteration, we find a minimally envied independent subset $I \subseteq U$ and swap it with a color class $C_i$ such that $i \in [k]$ envies $I$. We repeat this step until no independent subset of the uncolored vertices $U$ is envied by any $j \in [k]$. This procedure terminates since $\sum_{i = 1}^k w(C_i)$ strictly increases in every iteration. Finally, we complete the coloring by iteratively assigning each uncolored vertex $u \in U$ to a color class containing none of its neighbors. Such a color class always exists for every $u$, since $k \geq \Delta + 1$.
 
Let $(C_1, \dots, C_k)$ be the coloring obtained from this swapping procedure. It is a $k$-coloring of the subgraph $G[V']$, where $V' \coloneqq  \cup_{i=1}^k C_i \subseteq V$; it is indeed a valid coloring since, in each iteration, we swap in a subset $I$ that is independent in $G$. Finally, let $(C^*_1, \dots, C^*_k)$ denote the coloring of $G$ obtained by completing $(C_1, \ldots, C_k)$.

To prove that $(C^*_1, \dots, C^*_k)$ is $2$-\EQo, fix any two color classes $C^*_i$ and $C^*_j$. We have $C_i \subseteq C^*_i$ and $C_j \subseteq C^*_j$. Furthermore, there exists a $\bar{v} \in C_i$ such that $w(C_i \setminus \{\bar v\}) \le w(C_j)$; this follows because $C_i$ was chosen as a minimally envied subset when it was assigned. Moreover, $C^*_i \setminus C_i \subseteq C^*_i$ is an independent subset in $G$ and $C^*_i \setminus C_i$ is contained in $U = V \setminus \left( \cup_{a=1}^k C_a \right)$. Hence, by the termination condition of the algorithm's first while-loop, $w(C_j) \geq w(C^*_i \setminus C_i)$. Combining these observations, we obtain 
\begin{align}
         w(C^*_i \setminus \{\bar v\}) \le w(C_i \setminus \{\bar v\}) + w(C^*_i \setminus C_i) \le w(C_j) + w(C_j) \le 2w(C_j) \le 2w(C^*_j) \label{ineq:subadd} 
\end{align}
Therefore, $(C^*_1, \dots, C^*_k)$ is $2$-\EQo and the theorem stands proved. 
\end{proof}

\begin{remark}
\Cref{thm:delta-plus-one} extends to settings in which the weights assigned to vertex subsets are subadditive rather than additive. Specifically, for any subadditive function $w: 2^V \mapsto \mathbb{R}_{\geq 0}$ and any $k \geq \Delta + 1$, there exists a $k$-coloring $(C_1, \ldots, C_k)$ that satisfies the $2$-\EQo condition under $w$, i.e., for every $i,j \in [k]$, there exists a vertex $v \in C_i$ such that $w(C_i \setminus {v}) \leq 2 w(C_j)$. The proof of this extension is identical to that of \Cref{thm:delta-plus-one}; in particular, inequality~(\ref{ineq:subadd}) continues to hold even when $w$ is subadditive.     
\end{remark}

We complement \Cref{thm:delta-plus-one} with the following lower bound. 
\thmlowerbound*
\begin{proof}
The instance consists of a complete bipartite graph $K_{\Delta, \Delta}$ along with a separate clique $K_{\Delta + 1}$\footnote{This instance is similar to the lower bound instance of \cite{Hummel2022ConflictingItems} who also use complete bipartite graphs. Our improved lower bound is achieved using a careful setting of weights and a slightly more involved parity-based argument.}. Let $L$ and $R$ denote the two parts   of $K_{\Delta, \Delta}$. Each vertex in $L$ has weight $2$, and each vertex in $R$ has weight $\sqrt{2}$. All vertices in the clique $K_{\Delta + 1}$ have weight $0$. 

Assume towards a contradiction that there exists a $k$-coloring $(C_1, \ldots, C_k)$ with $k < \frac{3\Delta}{2}$ which is $\alpha$-\EQo for some $\alpha < \sqrt{2}$. Note that $K_{\Delta + 1}$ cannot be colored with fewer than $\Delta + 1$ colors, and hence $k > \Delta$.

Let $N \coloneqq \left\{i \in [k] \, |\, C_i \cap R \ne \emptyset \right\}$ be the set of color classes that contain vertices from $R$. Also, write $N^c = [k] \setminus N$. Note that $\sum_{i \in R} |C_i \cap R| = |R| = \Delta$ and $\Delta$ is odd. Hence, there must be at least one $i \in N$ for which $|C_i \cap R|$ is odd; let $i^*$ be one such index. 

For the independent set $C_{i^*}$, we have $C_{i^*} \cap R \neq \emptyset$, and hence $C_{i^*}$ contains no vertices from $L$. Furthermore, considering the weights of the vertices in the graph we obtain $w\left(C_{i^*} \right) = \sqrt{2} \ |C_{i^*} \cap R|$. We consider two exhaustive cases, depending on whether  $\left|C_{i^*} \cap R\right|= 1$ or $\left|C_{i^*} \cap R\right|$ is an odd number at least $3$. 

\noindent
\textbf{Case 1:} $|C_{i^*} \cap R| \ge 3$. In this case, even after removing any vertex $\overline{v} \in C_{i^*}$, we have $w \left(C_{i^*} \setminus \{ \overline{v} \} \right) \geq 2\sqrt{2}$. Hence, for the coloring to be $\alpha$-\EQo (for some $\alpha < \sqrt{2}$), all other color classes must have weight strictly greater than two. This implies that all color classes in $N$ must have at least two vertices from $R$, and all color classes in $N^c$ must have at least two vertices from $L$.
These two observations together imply $k \le \Delta$ which is less than the chromatic number of the graph. 

\noindent
\textbf{Case 2:} $|C_{i^*} \cap R| = 1$. In this case, we have $w(C_{i*}) = \sqrt{2}$ and $\alpha w(C_{i^*}) < 2$ (when $\alpha < \sqrt{2}$). If there is some color class $C_j$ that contains at least two vertices in $L$, then $w(C_j \setminus \{ \overline{v} \}) \ge 2$ for any $\overline{v} \in C_j$. Therefore, $w(C_j \setminus \{ \overline{v} \}) > \alpha w(C_{i^*})$, which is a contradiction to the coloring being $\alpha$-\EQo. This implies all color classes in $N^c$ must have at most one vertex from $L$; mathematically, $|N^c| \ge \Delta$. 

In addition, from Case 1, we assume that for all $i \in N$, we have $|C_i \cap R| \le 2$. This implies $|N| \ge \frac{\Delta}{2}$. Combining the two observations $|N^c| \ge \Delta$ and $|N| \ge \frac{\Delta}{2}$, we obtain $k \ge \frac{3\Delta}{2}$, a contradiction. \\

Overall, the case analysis shows that, for any $k < \frac{3 \Delta}{2}$, the graph does not admit a $k$-coloring that is $\alpha$-\EQo for some $\alpha < \sqrt{2}$. This completes the proof. 
\end{proof}

\section{Partition Equitability}\label{sec:partition-equitability}
This section establishes that for any given partition $(V_1, \ldots, V_d)$ of the vertex set $V$, and for any number of colors $k$ sufficiently large relative to $d$, there always exists a $k$-coloring $(C_1, \ldots, C_k)$ that is equitable with respect to the partition: $|C_i \cap V_j| \in \left \{ \left \lfloor  \frac{|V_j|}{k} \right \rfloor, \left \lceil  \frac{|V_j|}{k} \right \rceil \right \}$, for each $i \in [k]$ and $j \in [d]$. We will subsequently use this result to establish equitability under weights.

\thmpartitionequitability*

Our proof will use the following  lemma for extending colorings. 

\begin{lemma}\label{lem:atmost-one-addition}
Given a graph $G = (V, E)$ with maximum degree $\Delta$ and an integer $k \geq \Delta +1$, let vertex subset $V' \subset V$ be such that $|V\setminus V'| \leq k - \Delta$, and let $(C_1, \dots, C_k)$ be a coloring of the subgraph $G[V']$. Then there exists a coloring $(\widetilde{C}_1, \dots, \widetilde{C}_k)$ of $G$ such that  $C_i \subseteq \widetilde{C}_i$ and $|\widetilde{C}_i| \le |C_i| + 1$ for every $i \in [k]$. Furthermore, such a coloring $(\widetilde{C}_1, \dots, \widetilde{C}_k)$ can computed in polynomial time given the coloring $(C_1, \dots, C_k)$.
\end{lemma}
\begin{proof}
Each vertex $v \in V \setminus V'$ has a neighbor in at most $\Delta$ color classes among $(C_1, \dots, C_k)$, and hence can be assigned to at least $k-\Delta$ color classes without violating their independence. Since $|V \setminus V'| \le k - \Delta$, there exists a one-to-one mapping (matching) from the vertices in $V \setminus V'$ to the color classes $(C_1, \dots, C_k)$ such that each vertex  $v \in V \setminus V'$ can be included in the color class its mapped to without violating independence. Assigning vertices to color classes according to this mapping gives us a coloring $(\widetilde{C}_1, \dots, \widetilde{C}_k)$ of $G$ that satisfies the lemma's conditions. Additionally, such a mapping can be computed via a maximum matching algorithm, implying that $(\widetilde{C}_1, \dots, \widetilde{C}_k)$ can be computed in polynomial time.
\end{proof}
\subsection{Proof of \Cref{thm:partition-equitability}}
Our proof proceeds as follows: we first establish the existence of a partial coloring $(C_1, \ldots, C_k)$ of $V_1$ that is equitable with respect to this subset, i.e., the partial coloring satisfies $|C_i \cap V_1| \in \left \{ \left \lfloor  \frac{|V_1|}{k} \right \rfloor, \left \lceil  \frac{|V_1|}{k} \right \rceil \right \}$ for each $i \in [k]$. Then, we extend the partial coloring to one that is also equitable for $V_2$ without changing the coloring of the vertices in $V_1$. Subsequently we proceed to $V_3$, and so on.

We define some notation for the proof. Let $V_{[t]}$ denote the set $\bigcup_{j = 1}^t V_j$ and, for any graph $G'$, write $\cal C_k(G')$ to denote the set of all colorings of $G'$ with (exactly) $k$ colors. 

In addition, for each index $t \in [d]$, let $\cal E_t$ denote the set of colorings in $\cal{C}_k(G[V_{[t]}])$ that are 
\begin{inparaenum}[(a)]
    \item equitable with respect to the first $t$ subsets $V_1, V_2, \ldots, V_t$, and
    \item each $C_i$ has size at most $\left \lceil \frac{4|V_{[t]}|}{k} \right \rceil$.
\end{inparaenum}
Formally,  
\begin{align*}
    \cal{E}_t \coloneqq  \left \{(C_1, \ldots, C_k) \in \cal{C}_k(G[V_{[t]}]) \, \middle | \, \text{for each $i \in [k]$}, \ |C_i \cap V_j| \in \left \{ \left \lfloor  \frac{|V_j|}{k} \right \rfloor, \left \lceil  \frac{|V_j|}{k} \right \rceil \right \} \text{ for each } j \in [t] \right .\\
    \left. \text{ and } |C_i| \le \left \lceil \frac{4|V_{[t]}|}{k} \right \rceil \right \}. 
\end{align*}
Similarly, let $\cal F_t$ denote the set of colorings in $\cal{C}_k(G[V_{[t]}])$ that are
\begin{inparaenum}[(a)]
    \item equitable with respect to the first $(t-1)$ subsets $V_1, V_2, \ldots, V_{t-1}$, and
    \item each set $\left( C_i \cap V_{[t-1]} \right)$ has size at most $\left \lceil \frac{4|V_{[t-1]}|}{k} \right \rceil$.
\end{inparaenum}
Specifically, 
\begin{align*}
    \cal{F}_t \coloneqq \left \{(C_1, \ldots, C_k) \in \cal{C}_k(G[V_{[t]}]) \, \middle | \, \text{for each $i \in [k]$}, \ |C_i \cap V_j| \in \left \{ \left \lfloor  \frac{|V_j|}{k} \right \rfloor, \left \lceil  \frac{|V_j|}{k} \right \rceil \right \} \text{ for each } j \in [t-1] \right .
    \\ \left .\text{ and } |C_i \cap [V_{[t-1]}]| \le \left \lceil \frac{4|V_{[t-1]}|}{k} \right \rceil \right \}. 
\end{align*}
Note that both $\cal{E}_t$ and $\cal{F}_t$ consist of $k$-colorings of $G[V_{[t]}]$. 

Our proof has two inductive steps over the indices $t$: \begin{inparaenum}[(a)]
    \item We first show that if $\cal E_{t-1}$ is nonempty, then $\cal F_t$ is also nonempty; and 
    \item we then show that if $\cal F_t$ is nonempty, then $\cal E_t$ is nonempty.
\end{inparaenum}

For the base case, $\cal E_1$ is nonempty by the Hajnal-Szemer\'{e}di Theorem~\cite{HajnalSzemeredi1970}; note that for $\cal E_1$, any coloring that is equitable with respect to $V_1 = V_{[1]}$ satisfies the cardinality constraint 
$|C_i \cap V_{[1]}| \leq \left\lceil \frac{4 |V_{[1]}|}{k} \right\rceil$.

Moreover, the implication $\cal{E}_d \neq \emptyset$, established inductively via steps (a) and (b), ensures that, as desired, there exists a $k$-coloring of the given graph $G$ that is equitable with respect to all $d$ subsets $V_1, \ldots, V_d$.  Therefore, steps (a) and (b) complete the proof. 

The analysis of the steps is divided into two cases based on whether $|V_t| > \frac{k}{2}$. When $|V_t| \le \frac{k}{2}$, the following lemma shows that $\cal E_t$ is nonempty; this lemma notably does not use $\cal F_t$. 

\begin{lemma}\label{lem:vt-small}
For any index $1 < t \leq d$, if $|V_t| \le \frac{k}{2}$ and $\cal E_{t-1} \ne \emptyset$, then $\cal E_{t} \ne \emptyset$.
\end{lemma}
\begin{proof}
Let $(C_1, \dots, C_k)$ be a coloring in $\cal E_{t-1}$. Re-index the color classes so that $|C_1| \le |C_2| \le \dots \le |C_k|$. 

If $|V_{[t]}| = \left|V_{[t-1]} \cup V_t \right| < k$, then the first $|V_t|$ color classes $C_1, \dots, C_{|V_{t}|}$ are empty. We assign each vertex from $V_t$ to these color classes such that each class receives exactly one vertex. This yields a coloring of the graph $G[V_{[t]}]$ that lies in $\cal E_{t}$. 

Hence, for the remainder of the proof, assume that that $|V_{[t]}| \geq k$. The parameter $\eta$ (as specified in the theorem statement) satisfies $\eta \ge 1$, and hence $k \ge 6\Delta$. Define $k' \coloneq |V_t| + \Delta$. Since $|V_t| \leq k/2$ and $\Delta \leq k/6$, we have  
$k' = |V_t|  + \Delta \leq \frac{2k}{3}$. 

Moreover, by applying Lemma \ref{lem:atmost-one-addition} with the set $V_t$ $(|V_t| = k' - \Delta)$, we can assign the vertices of $V_t$ to color classes $C_1, \dots, C_{k'}$ such that each color class here receives at most one vertex from $V_t$. This extension yields  a coloring $(\widetilde{C}_1, \dots, \widetilde{C}_{k})$ of the graph $G[V_{[t]}]$ that satisfies the equitability constraints of $\cal E_{t}$; in particular, $\left| \widetilde{C}_i \cap V_j \right| \in \{0,1\}$ for every $i \in [k]$ and $j \in [t]$.\footnote{Recall that in the current lemma $|V_j| \leq k/2$ for every $j \in [t]$. Hence, $0 \leq |V_j|/k \leq 1$.}

We next show that $(\widetilde{C}_1, \dots, \widetilde{C}_{k})$ also satisfies the cardinality constraint required for membership in $\cal E_t$. Since $\sum_{i = 1}^k |C_i| = V_{[t-1]}$ and the color classes are indexed in nondecreasing order of size, we have 
\begin{align}
    |V_{[t-1]}| \ge (k - k') |C_{k'}|  \ge \frac{k}{3} |C_{k'}|  \tag{since $k' \leq \frac{2k}{3}$} 
\end{align}
Therefore, for all indices $1 \le i \le k'$,   
\begin{align*}
    |\widetilde{C}_i| \le |C_i| + 1 \le |C_{k'}| + 1 \le \frac{3 |V_{[t-1]}|}{k} + 1 \leq \frac{3 |V_{[t-1]}|}{k}  + \frac{|V_{[t]}|}{k} \le \frac{4|V_{[t]}|}{k}.
\end{align*}
Here, the penultimate inequality uses $|V_{[t]}| \geq k$. For the remaining indices $k' + 1 \leq i \leq k$, since the color classes are unchanged, we have  
\begin{align*}
    |\widetilde{C}_i| = |C_i| \le \left \lceil \frac{4|V_{[t-1]}|}{k} \right \rceil \le \left \lceil \frac{4|V_{[t]}|}{k} \right \rceil.
\end{align*}

Overall, we conclude that the the coloring  $(\widetilde{C}_1, \dots, \widetilde{C}_k)$ lies in $\cal E_t$, and hence $\cal E_t \ne \emptyset$. This completes the proof of the lemma. 
\end{proof}

\begin{lemma}\label{lem:ft-nonempty}
For any integer $1 < t \leq d$, if $\cal E_{t-1} \ne \emptyset$, then $\cal F_{t} \ne \emptyset$.
\end{lemma}
\begin{proof}
Start with any coloring $(C_1, \ldots, C_k) \in \cal E_{t-1}$. 
We extend $(C_1, \ldots, C_k)$ to a $k$-coloring of $G[V_{[t]}]$ by assigning colors to the vertices $v \in V_t$. In particular, we add each vertex $v$ to a color class that does not contain any of its neighbors in $G[V_{[t]}]$. This is always possible since $k \ge 6\Delta$ and $v$ has at most $\Delta$ neighbors in $G[V_{[t]}]$. The resulting coloring lies in $\cal F_{t}$. 
\end{proof}

\begin{lemma}\label{lem:et-nonempty}
For any integer $1 \leq t \leq d$, if $|V_t| \ge \frac{k}{2}$ and $\cal F_{t} \ne \emptyset$, then $\cal E_{t} \ne \emptyset$.
\end{lemma}
\begin{proof}
We show that there exists a coloring in $\cal F_t$ that is also in $\cal E_t$. Specifically, select a coloring $(C_1, \ldots, C_k) \in \cal F_t$ that minimizes $\sum_{i \in [k]} |C_i \cap V_t|^2$.

If the selected coloring satisfies the $t^{\text{th}}$ equitability constraint for $\cal E_t$ (i.e., $|C_i \cap V_t| \in \left \{ \left \lfloor  \frac{|V_t|}{k} \right \rfloor, \left \lceil  \frac{|V_t|}{k} \right \rceil \right \}$ for every $i \in [k]$), then 
$(C_1, \ldots, C_k)$ satisfies $\cal E_t$'s cardinality constraint as well:
\begin{align}
    |C_i| & = |C_i \cap V_{[t-1]}| + |C_i \cap V_t| \nonumber \\ 
    & \le  |C_i \cap V_{[t-1]}| + \left \lceil \frac{|V_{t}|}{k} \right \rceil \nonumber \\
    & \le \left \lceil \frac{4|V_{[t-1]}|}{k} \right \rceil + \left \lceil \frac{|V_{t}|}{k} \right \rceil \tag{since $(C_1, \ldots, C_k) \in \cal F_t$} \\ 
    & \le \left( \frac{4|V_{[t-1]}|}{k} + 1 \right) + \frac{2|V_{t}|}{k} \tag{since $|V_t| \geq k/2$} \\ 
    & \le \frac{4|V_{[t]}|}{k} \label{eq:cardinality-constraints}
\end{align}
Hence, to establish that the selected coloring $(C_1, \ldots, C_k)$ lies in $\cal E_t$, it suffices to show that $|C_i \cap V_t| \in \left \{ \left \lfloor  \frac{|V_t|}{k} \right \rfloor, \left \lceil  \frac{|V_t|}{k} \right \rceil \right \}$ for every $i \in [k]$. Assume, towards a contradiction, that this does not hold, and that $|C_1 \cap V_t| \ge |C_2 \cap V_t| \ge \ldots \ge |C_k \cap V_t| = \ell$ with $|C_1 \cap V_t| \ge \ell + 2$. Indeed, if $|C_1 \cap V_t| \leq \ell + 1$, then both $|C_1 \cap V_t|$ and $|C_k \cap V_t|$ must lie in $\left \{ \left \lfloor  \frac{|V_t|}{k} \right \rfloor, \left \lceil  \frac{|V_t|}{k} \right \rceil \right \}$. Now, write $S \subset [k]$ to denote the set of indices $i \in [k]$ such that $|C_i \cap V_t| = \ell$. That is, $S$ contains the color classes with the smallest intersection with $V_t$. 

Define the directed graph $H = ([k], E')$, where there is a directed edge from $i \in [k]$ to $i' \in [k]$ if there exists a vertex in $C_i \cap V_t$ that can be moved to $C_{i'}$ without violating the independence of $C_{i'}$ in the subgraph $G[V_{[t]}]$. Also, let $R \subset [k]$ denote the set of vertices of $H$ that have a directed path to some vertex in $S$; include $S$ into $R$ as well. Formally,
\begin{align*}
    R \coloneqq  \left\{i \in [k] \mid \text{ there is a directed path from $i$ to $S$ in $H$} \right\} \cup S.
\end{align*}

We establish the following two claims about $R$. 

\begin{claim}\label{claim:r-classes-size}
If $i \in R$, then $|C_i \cap V_t| \le \ell + 1$.
\end{claim}
\begin{proof}
Assume, towards a contradiction, that for $i \in R$ we have $|C_i \cap V_t| \ge \ell + 2$. Since there is a path from $i$ to some $i'$ in $S$, we can move vertices from $V_t$ along this path to reduce $|C_i \cap V_t|$ by $1$ and increase $|C_{i'} \cap V_t|$ by $1$, while retaining the property that $(C_1, \ldots, C_k)$ is still a coloring in $\cal F_t$. 

Specifically, let $(i, i_1, \ldots, i_r, i')$ be the shortest (directed) path from $i$ to $i'$ in $H$. For each edge along this path, say $(i_x, i_{x+1})$, there is a vertex $q$ in $C_{i_x} \cap V_t$ that can be moved to $C_{i_{x+1}}$ without violating the independence of $C_{i_{x+1}}$ (in $G[V_{[t]}]$). We move this vertex $q$ from $C_{i_x}$ to $C_{i_{x+1}}$ and repeat this step for every edge in the path. Let $(C'_1, \ldots, C'_k)$ denote the resulting coloring. Since we chose the shortest path, each color class has at most one vertex added and at most one vertex removed. Therefore, $(C'_1, \ldots, C'_k)$ is still a coloring of $G[V_{[t]}]$.

Note that only the vertices in $V_t$ change color classes, and hence $(C'_1, \ldots, C'_k)$ remains in $\cal F_{t}$. Furthermore, since $|C_i \cap V_t| \ge |C_{i'} \cap V_t| + 2$, the new coloring $(C'_1, \ldots, C'_k)$ satisfies $\sum_{i \in [k]} |C'_i \cap V_t|^2 < \sum_{i \in [k]} |C_i \cap V_t|^2$. This contradicts our choice of coloring $(C_1, \ldots, C_k)$ as a coloring in $\cal F_t$ that minimizes $\sum_{i \in [k]}|C_i \cap V_t|^2$. Hence, by way of contradiction, the claim follows. 
\end{proof}

\begin{claim}\label{claim:r-upper-bound}
$|R| \le \Delta$.   
\end{claim}
\begin{proof}
Assume, towards a contradiction, that $|R| \ge \Delta + 1$. Let $v$ be any vertex in $C_1 \cap V_t$; recall that $|C_1 \cap V_t| \geq \ell +2 > 0$, and hence such a vertex exists. 

Now, since the degree of vertex $v$ in $G$ (and hence in $G[V_{[t]}]$) is at most $\Delta$, there is at least one color class in $R$ that does not contain any neighbor of $v$. Therefore, $v \in C_1 \cap V_t$ can be moved to that color class, implying that color class $1$ has an outgoing edge in $H$ to a vertex in $R$. Hence, $1 \in R$. This containment, however, contradicts Claim \ref{claim:r-classes-size}, since $|C_1 \cap V_t| \geq \ell+2$. Hence, by way of contradiction, the claim stands proved. 
\end{proof}

For the color class $C_k$ in the coloring $(C_1, \ldots, C_k) \in \cal{F}_t$ we have 
\begin{align}
    |C_k| & = |C_k \cap V_{[t-1]}| + |C_k \cap V_t| \nonumber \\ 
    & \le \left \lceil \frac{4 |V_{[t-1]}|}{k} \right \rceil + \ell \tag{since $(C_1,\ldots,C_k) \in \cal{F}_t$ and $|C_k \cap V_t| = \ell$} \\ 
    & \le \left( \frac{4 |V_{[t-1]}|}{k} + 1 \right) + \frac{|V_t|}{k} \nonumber \\ 
    & < \frac{4|V_{[t]}|}{k} \tag{since $|V_t| \geq k/2$} \\
    & \le \frac{4\eta|V_{t}|}{k} \label{ineq:size-of-Ck}
\end{align}
The last inequality follows from the definition of $\eta$ in the theorem statement.

Furthermore, we lower bound the number of vertices in $\bigcup_{i \notin R} (C_i \cap V_t)$ as follows 
\begin{align}
\left| \bigcup_{i \notin R} ( C_i \cap V_t ) \right| = \sum_{i \notin R} \left| C_i \cap V_t \right| = |V_t| - \sum_{i \in R} |C_i \cap V_t| \underset{\text{\Cref{claim:r-classes-size}}}{\geq} |V_t| - |R| \frac{|V_t|}{k} = (k - |R|)\frac{|V_t|}{k} \label{ineq:comp-size}
\end{align}
Recall that $k > (4\eta + 2) \Delta$ and $|R| \leq \Delta$ (\Cref{claim:r-upper-bound}). Hence, inequality \eqref{ineq:comp-size} extends to 
\begin{align}
\left| \bigcup_{i \notin R} ( C_i \cap V_t ) \right| \ > 4\eta \Delta \frac{|V_t|}{k} \underset{\text{via ineq.~\eqref{ineq:size-of-Ck}}}{\geq} \Delta |C_k|
\label{ineq:comparison}
\end{align}
Furthermore, note that each vertex in $\bigcup_{i \notin R} (C_i \cap V_t)$ has an edge to some vertex in $C_k$; otherwise, there would be an edge from some $i \notin R$ to $k$ in the graph $H$, contradicting the fact that $i \notin R$. Therefore, considering the total number of edges incident to the vertices in $C_k$, we obtain $\Delta |C_k| \geq \left| \bigcup_{i \notin R} ( C_i \cap V_t ) \right|$. This bound, however, contradicts inequality~\eqref{ineq:comparison}.

Hence, the assumption that in the selected coloring $(C_1, \ldots, C_k) \in \cal F_t$ we have $|C_1 \cap V_t| \ge |C_k \cap V_t| + 2$ does not hold. Therefore, the coloring $(C_1, \ldots, C_k)$ is equitable with respect to $V_t$ as well and it is satisfies the cardinality constraints of $\cal E_t$ due to inequality \eqref{eq:cardinality-constraints}; this implies $(C_1, \ldots, C_k) \in \cal E_t$. Therefore $\cal E_t \neq \emptyset$, completing the proof of the lemma.
\end{proof}

Inductively applying Lemmas~\ref{lem:vt-small},~\ref{lem:ft-nonempty}, and~\ref{lem:et-nonempty} shows that $\cal E_d$ is nonempty, which establishes the existence of the desired coloring that is equitable with respect to the given partition $(V_1, \ldots, V_d)$.

This existence proof also yields a polynomial-time algorithm for computing an equitable coloring. Specifically, the Lemmas~\ref{lem:vt-small},~\ref{lem:ft-nonempty}, and~\ref{lem:et-nonempty} can be converted into efficient algorithms. This conversion is straightforward for Lemma~\ref{lem:ft-nonempty}, since it constructs a coloring in $\cal{F}_t$ by arbitrarily assigning valid colors to the vertices in $V_t$. Lemma~\ref{lem:vt-small} invokes Lemma~\ref{lem:atmost-one-addition} to find a valid coloring of the vertices in $V_t$; as stated in Lemma~\ref{lem:atmost-one-addition}, this coloring can be found in polynomial time.

To obtain an algorithm for Lemma~\ref{lem:et-nonempty}, it suffices to construct a coloring in $\cal{E}_t$ given one in $\cal{F}_t$. Notably, the proof of Lemma~\ref{lem:et-nonempty} shows that any coloring in $\cal{F}_t$ that satisfies Claim~\ref{claim:r-classes-size} lies in $\cal{E}_t$. Furthermore, for any given coloring $(C_1, \ldots, C_k) \in \cal{F}_t$ that violates Claim~\ref{claim:r-classes-size}, the proof of the claim describes a polynomial-time path augmentation procedure that constructs a new coloring $(C'_1, \ldots, C'_k)$ such that $\sum_{i \in [k]} |C'_i \cap V_t|^2 < \sum_{i \in [k]} |C_i \cap V_t|^2$.

Therefore, repeatedly applying this path augmentation procedure eventually yields a coloring $(C'_1, \ldots, C'_k)$ that satisfies Claim~\ref{claim:r-classes-size}. The procedure can be applied only a polynomial number of times since $\sum_{i \in [k]} |C_i \cap V_t|^2$ is an integer between $1$ and $n^3$. Overall, we obtain a polynomial-time algorithm for finding an equitable coloring with respect to the given partition $(V_1, \ldots, V_d)$. This completes the proof of the theorem.


\subsection{Implications of \Cref{thm:partition-equitability}}
The corollary below shows that in any graph with maximum degree $\Delta$, with $k=O(d \Delta)$ colors 
we can achieve equitability for any partition with $d$ parts.


\begin{corr}\label{cor:partition-equitability-multiple}
Given a graph $G = (V, E)$ with maximum degree $\Delta \in \mathbb{Z}_+$, let $(V_1,\ldots,V_d)$ be any partition of the vertex set $V$ satisfying $|V_1| \le |V_2| \le \ldots \le |V_d|$. Also, let $k \in \mathbb{Z}_{+}$ be an integer such that $k \ge (4d + 2)\Delta$ and $|V_j|$ is a multiple of $k$ for each $j \in [d]$. Then $G$ admits a $k$-coloring $(C_1,\ldots, C_k)$ satisfying $|C_i \cap V_j| =\frac{|V_j|}{k}$ for all $i \in [k]$ and $j \in [d]$. Furthermore, such a coloring can be computed in polynomial time.
\end{corr}
\begin{proof}
Since $|V_1| \leq \ldots \leq |V_d|$, the parameter $\eta$ in the statement of \Cref{thm:partition-equitability} satisfies 
\begin{align}
\eta = \max_{t \in [d]} \left(  \frac{\sum_{j = 1}^{t} |V_j|}{|V_t|} \right) \leq \max_{t \in [d]} t = d
\end{align}
In addition, since each $|V_j|$ is a multiple of $k$, we have $\left \lfloor  \frac{|V_t|}{k} \right \rfloor = \left \lceil  \frac{|V_t|}{k} \right \rceil = \frac{|V_t|}{k}$. 
 
Therefore, we can instantiate \Cref{thm:partition-equitability}, with $k \geq (4d+2)\Delta \geq (4 \eta + 2) \Delta$, to obtain a $k$-coloring $(C_1,\ldots, C_k)$ with $|C_i \cap V_j| =\frac{|V_j|}{k}$ for all $i \in [k]$ and $j \in [d]$. This completes the proof of the corollary.  
\end{proof}

We also present a method for extending a given partial coloring into a complete coloring that satisfies certain desirable properties.

\begin{lemma}\label{lem:filling-up}
Let $G = (V, E, w)$ be a vertex-weighted graph and let $V' \subseteq V$. Suppose $(C'_1, \ldots, C'_k)$ is a $k$-coloring of $G[V']$ with $k \geq \Delta+1$. Then there exists a coloring $(C_1, \ldots, C_k)$ of $G$ such that, for all $i \in [k]$: 
\begin{enumerate}[(a)]
    \item $C'_i \subseteq C_i$ and
    \item $w(C_i \setminus C'_i) - \normalfont{\maxw}(C_i \setminus C'_i) \le \frac{w(V \setminus V')}{(k - \Delta)}$.
\end{enumerate}
Furthermore, such a coloring $(C_1, \ldots, C_k)$ can be computed in polynomial time.
\end{lemma}
\begin{proof}
Let $V \setminus V'= \{1, 2, \ldots, m\}$ and assume that $w_1 \ge w_2 \ge \ldots \ge w_m$. We can also assume, without loss of generality, that $(k-\Delta)$ divides $m$; otherwise, we add isolated dummy vertices of weight $0$ to $V \setminus V'$. We partition the set $V \setminus V'$ into subsets of size $(k - \Delta)$. Specifically, write $\widetilde{V}_j = \{(j-1)(k - \Delta) + 1, \ldots, j(k - \Delta)\}$ for each $j \in \left\{1, 2, \ldots, \frac{m}{k - \Delta} \right\}$.

We construct the desired coloring by inductively processing the subsets $\widetilde{V}_1, \ldots, \widetilde{V}_{m/(k-\Delta)}$. In particular, for indices $1 \leq j \leq \frac{m}{k - \Delta}$, we do the following: given a $k$-coloring $(\widetilde{C}_1, \ldots, \widetilde{C}_k)$ of the subgraph $G \left[V' \cup \bigcup_{t = 1}^{j-1} \widetilde{V}_t \right]$, we construct a coloring $(\widehat{C}_1, \ldots, \widehat{C}_k)$ of $G \left[V' \cup \bigcup_{t = 1}^{j} \widetilde{V}_t \right]$ such that $\widetilde{C}_i \subseteq \widehat{C}_i$ for all $i \in [k]$, and each vertex in $\widetilde{V}_j$ is assigned a different color. This can be done using Lemma \ref{lem:atmost-one-addition}. 

Starting with the given coloring $(C'_1, \ldots, C'_k)$ of $G[V']$ and repeatedly applying the above extension, we obtain the desired coloring $(C_1, \ldots, C_k)$ of the graph $G$. By construction, $C'_i \subseteq C_i$ for all $i \in [k]$, and hence property (a) in the lemma statement holds. In addition, note that $|C_i \cap \widetilde{V}_j| \leq 1$ for all $i \in [k]$ and $j \in \left[ \frac{m}{k - \Delta} \right]$. We will use this bound to establish property (b) for each $i \in [k]$: 
\begin{align*}
    w(C_i) & = w(C'_i) + \sum_{j = 1}^{\frac{m}{k-\Delta}} w \left(C_i \cap \widetilde{V}_j \right) \\
    &\le w(C'_i) + w \left(C'_i \cap \widetilde{V}_1 \right) + \sum_{j = 2}^{\frac{m}{k-\Delta}}\maxw \left( \widetilde{V}_j \right)  \tag{since $|C_i \cap \widetilde{V}_j| \leq 1$} \\
    &\le w(C'_i)  + \maxw(C'_i \setminus C_i) +  \sum_{j = 2}^{\frac{m}{k-\Delta}}\maxw \left( \widetilde{V}_j \right) \\ 
    &\le w(C'_i)  + \maxw(C'_i \setminus C_i) + \sum_{j = 2}^{\frac{m}{k-\Delta}} \frac{1}{k-\Delta} w \left(\widetilde{V}_{j-1} \right) \tag{Wts.~are indexed in non-increasing order}\\ 
    &\le w(C'_i) + \maxw(C'_i \setminus C_i) + \frac{w(V \setminus V')}{k-\Delta}.
\end{align*}
This gives us property (b). The coloring $(C_1, \dots, C_k)$ is constructed by applying Lemma \ref{lem:atmost-one-addition} at most $n$ times. Therefore, the coloring $(C_1, \dots, C_k)$ can be computed in polynomial time. This completes the proof.
\end{proof}

\section{\EQo under Weights}\label{sec:exact-eqo}
\label{section:exact-eqo}
This section develops two results on the existence of exact \EQo colorings (\Cref{defn:eqo}). The theorem below shows that for vertex-weighted graphs with $d \in \mathbb{Z}_+$ distinct weights, $(8d+10) \Delta$ colors suffice to achieve exact equitability. Furthermore, \Cref{cor:simple-round-robin} (stated below) shows that $4\sqrt{n \Delta}$ colors suffice for exact \EQo; note that this bound is independent of the number of distinct weights in the instance.

\thmdequitable*

\begin{proof}
We index the vertices $V = \{1, 2, \ldots, n\}$ in non-increasing order of weight: $w_1 \ge w_2 \ge \ldots \ge w_n$. We also assume that $k$ divides $n$; otherwise, we add a few isolated dummy vertices of weight $0$ so that $k$ divides $n$. Note that any \EQo coloring of this modified instance remains an \EQo coloring after removing the dummy vertices. Furthermore, this increases the number of distinct weights from $d$ to at most $d+1$. 

We start with an equipartition $(V'_1, \ldots, V'_{n/k})$ of the vertex set $V$ and modify it to reduce the number of parts. Specifically, set $V'_j = \{(j-1)k + 1, \ldots, jk\}$ for each $j \in [n/k]$. Note that the ordering of the weights implies $\minw(V'_j) \ge \maxw(V'_{j+1})$ for each $j \in [\frac{n}{k}-1]$.

We start with this equipartition, and apply the following {\em merging procedure} to reduce the number of parts: for any index $i$ such that all vertices in $V'_i$ and $V'_{i+1}$ have the same weight, merge the two subsets to create a new partition $(V''_1, \ldots, V''_{n/k-1})$. Specifically, for each $j \in [\frac{n}{k}-1]$ set 
\begin{align*}
    V''_j = 
    \begin{cases}
        V'_j & j < i \\
        V'_i \cup V'_{i + 1} & j = i \\
        V'_{j + 1} & j > i
    \end{cases}
\end{align*}
We repeat the procedure until no further merges are possible. Write $(\widetilde{V}_1, \ldots, \widetilde{V}_{\widetilde{d}})$ to denote the resulting partition. By construction, this partition also satisfies 
\begin{align}
\minw( \widetilde{V}_j) \ge \maxw( \widetilde{V}_{j+1}) \qquad \text{for each $j \in [\widetilde{d}-1]$} \label{ineq:wt-sort}
\end{align}

Furthermore, the merging ensures that $|\widetilde{V}_j|$ is divisible by $k$ for every $j$. The following claim shows that in the partition $(\widetilde{V}_1, \ldots, \widetilde{V}_{\widetilde{d}})$ the number of parts $\widetilde{d}$ is at most $2(d + 1)$; recall that $(d+1)$ is an upper bound on the number of distinct weights in the given graph.

\begin{claim}
\label{claim:d-dtilde}
$\widetilde{d} \leq 2(d + 1)$.
\end{claim}
\begin{proof}
The subsets in the partition $(\widetilde{V}_1, \ldots, \widetilde{V}_{\widetilde{d}})$ with odd indices have distinct maximum weights. In particular, for each integer  $1 \leq j  < \frac{\widetilde{d}}{2}$, we have $\maxw( \widetilde{V}_{2j-1}) >  \maxw( \widetilde{V}_{2j+1})$. Indeed, if $\maxw( \widetilde{V}_{2j-1}) =  \maxw( \widetilde{V}_{2j+1})$, then all vertices in the subsets $\widetilde{V}_{2j-1}$ and $\widetilde{V}_{2j}$ would have the same weight, and these two subsets would have been merged in the merging procedure. Since the total number of distinct weights is at most $(d+1)$ it follows that $\frac{\widetilde{d}}{2} \leq d+1$. This establishes the claim.   
\end{proof}

We now invoke \Cref{cor:partition-equitability-multiple} with graph $G=(V,E)$ and partition $(\widetilde{V}_1, \ldots, \widetilde{V}_{\widetilde{d}})$ to obtain coloring $(C_1, \ldots, C_k)$. Note that the corollary applies here since each $|\widetilde{V}_j|$ is a multiple of $k$, and the number of colors satisfies $k \geq (8d + 10)\Delta \geq (4 \widetilde{d}+2)\Delta$; the last inequality follows from \Cref{claim:d-dtilde}. We complete the proof next by showing that the obtained coloring $(C_1, \ldots, C_k)$ is \EQo with respect to the given weights.

Consider any two color classes $C_i$ and $C_{i'}$. \Cref{cor:partition-equitability-multiple} ensures that $|C_i \cap \widetilde{V}_j | = |C_{i'} \cap \widetilde{V}_j| = \frac{|\widetilde{V}_j|}{k}$ for every $j  \in [\widetilde{d}]$. Let $J = \{j_1, j_2, \ldots, j_t\}$, with $j_1 < j_2 <\ldots < j_t$, denote the indices of the subsets $\widetilde{V}_j$ for which $\frac{|\widetilde{V}_j|}{k} = 1$. Furthermore, let $\bar{v}$ be the unique element in $C_i \cap \widetilde{V}_{j_1}$. Note that any subset $\widetilde{V}_j$ with size greater than $k$ (i.e., $\frac{|\widetilde{V}_j|}{k} > 1$) must have been obtained by merging two subsets in the merging procedure described above. Hence, 
\begin{align}
\maxw(\widetilde{V}_j) = \minw(\widetilde{V}_j) \qquad \text{for each $j \notin J$} \label{ineq:max-min}
\end{align}
These observations imply that \EQo holds between color classes $C_i$ and $C_{i'}$: 
{\allowdisplaybreaks
\begin{align*}
    w(C_i \setminus \{ \bar{v} \}) &\le \sum_{j \in J \setminus \{j_1\}} \maxw(\widetilde{V}_j) + \sum_{j \notin J}  \frac{|\widetilde{V}_j|}{k} \ \maxw(\widetilde{V}_j) \\
    &= \sum_{j \in J \setminus \{j_1\}} \maxw(\widetilde{V}_j) + \sum_{j \notin J}  \frac{|\widetilde{V}_j|}{k} \ \minw(\widetilde{V}_j) \tag{via \eqref{ineq:max-min}} \\
    & \leq \sum_{j \in J \setminus \{j_t\}} \minw(\widetilde{V}_j) + \sum_{j \notin J}  \frac{|\widetilde{V}_j|}{k} \ \minw(\widetilde{V}_j) \tag{via \eqref{ineq:wt-sort}}\\
    &\le w(C_{i'}).
\end{align*}
}
This proves that the coloring $(C_1, \ldots, C_k)$ is \EQo. There are three steps used to construct the coloring $(C_1, \dots, C_k)$:
\begin{inparaenum}[(a)]
    \item create the initial partition $(V'_1, \dots, V'_{n/k})$, 
    \item apply the merging procedure, and
    \item invoke Corollary \ref{cor:partition-equitability-multiple}.
\end{inparaenum}
All three of these steps can be done in polynomial time. Therefore the \EQo coloring $(C_1, \ldots, C_k)$ can be computed in polynomial time
\end{proof}

We obtain another implication by applying \Cref{cor:partition-equitability-multiple} directly to the equipartition $(V'_1, \ldots, V'_{n/k})$ of the vertex set $V$; that is, we do not use the merging procedure in the above proof. In this case, we require the number of colors to satisfy 
\begin{align*}
    k \ge \left ( 4\left \lceil \frac{n}{k} \right \rceil + 2 \right )\Delta.
\end{align*}
Formally,  

\begin{corr}\label{cor:simple-round-robin}
Let $G = (V, E, w)$ be a vertex-weighted graph with maximum degree $\Delta$. Then $G$ admits an \EQo $k$-coloring for every integer $k \ge 4\sqrt{n\Delta}$. Furthermore, such a coloring can be computed in polynomial time.
\end{corr}

\section{$(1+\varepsilon)$-\EQo under Weights}
\label{section:apx-eqo}
This section shows that for any $\varepsilon >0$ and any vertex-weighted graph with maximum degree $\Delta$, we can achieve equitability up to one vertex within a $(1+\varepsilon)$ multiplicative factor using $k = \widetilde{O}\left( \Delta/\varepsilon^2 \right)$ colors. Notably, here for $(1+\varepsilon)$-\EQo (\Cref{defn:eqo}), and in contrast to the bounds obtained in \Cref{section:exact-eqo}, $k$ does not explicitly depend on either the number of distinct weights or the number of vertices.

\thmepsequitable*

\subsection{Proof of \Cref{thm:eps-equitable}}\label{sec:apx-eqo-proof}

We prove this theorem by invoking \Cref{thm:partition-equitability}. To this end, we construct a partition such that (i) \EQo with respect to the partition implies  $(1+\varepsilon)$-\EQo with respect to the given vertex weights, and (ii) the parameter $\eta$ (as defined in \Cref{thm:partition-equitability}) for the partition is $\widetilde{O}(1/\varepsilon^2)$. Applying \Cref{thm:partition-equitability} then yields a $(1+\varepsilon)$-approximate \EQo $k$-coloring with $k = \widetilde{O} \left( \Delta/\varepsilon^2 \right)$.

We index the vertices $V=\{1,2, \ldots, n\}$ such that  $w_1 \ge w_2 \ge \ldots \ge w_n$. Also assume, without loss of generality, that $n \ge \frac{2}{\vare^2}$; otherwise, we can apply \Cref{cor:simple-round-robin}.

Our construction of the desired partition is parameterized by two integers: 
\begin{align}
    x \coloneqq \left\lceil \frac{1}{\vare^2} \right\rceil \quad \text{ and } \quad y \coloneqq \left \lceil \frac{1 + \vare}{\vare}\ln{\frac1{\vare}} + 1 \right \rceil
    \label{eq:defnxy}
\end{align}
We create the first $x$ vertex subsets (parts), $T_1, \ldots, T_x$, of the partition as size-$k$ subsets consisting of the $kx$ highest-weight vertices. Specifically, for each index $j \in [x]$, let $T_j = \{k(j-1) + 1, \ldots, kj\}$. 

We normalize the weights so that $w_{kx+1}=1$; note that this normalization does not affect the approximate \EQo guarantee. Under this normalization and the ordering of the weights, the remaining vertices, $\{kx +1, \ldots, n\}$, have weights at most $1$.   

We partition the remaining vertices dyadically: create  $\ell \in \mathbb{Z}_+$ subsets $B_1, B_2, \ldots, B_\ell$, where $B_j$ contains the vertices in $\{kx +1, \ldots, n\}$ with weights in $\left (\frac{1}{(1 + \vare)^{j}}, \frac{1}{(1 + \vare)^{j-1}} \right ]$.
With sufficiently large $\ell$, each vertex in $\{kx +1, \ldots, n\}$ belongs to exactly one subset $B_j$; some subsets may be empty. 

We start with the partition $(T_1, \ldots, T_x, B_1, \ldots, B_{\ell})$ and process it before invoking \Cref{thm:partition-equitability}. This processing will ensure that \begin{inparaenum}[(a)]
    \item each $|B_i|$ is a multiple of $k$,\footnote{By construction, we already have that $|T_j| = k$ for every $j \in [x]$.}
    and 
    \item the resulting parameter $\eta = \widetilde{O}(1/\varepsilon^2)$.
\end{inparaenum}

We obtain both properties (a) and (b) by setting aside a few vertices from some of the subsets $B_j$. We will show that the total weight of these deferred vertices is small enough, and we can apply Lemma~\ref{lem:filling-up} to color them later equitably.

Formally, we create a deferred set $D$, initialized to $\emptyset$. When we modify the subsets $B_j$, we move some vertices from $B_j$ to $D$; we refer to this operation as deferring a vertex.

To ensure property (a), i.e., that each $|B_j|$ is divisible by $k$, we can defer fewer than $k$ vertices from each $B_j$ so that the updated cardinality of $B_j$ is divisible by $k$. Hence, after moving these vertices into $D$, each $|B_j|$ is divisible by $k$. 
We can upper bound the total weight of the deferred vertices here---denoted as $W'_{\text{def}}$---as follows: 
\begin{align}
    W'_{\text{def}} \le \sum_{j \geq 1} k \  \frac{1}{(1+\vare)^{j-1}} \le k \frac{(1+\vare)}{\vare} \label{eq:discard-multiple}
\end{align}

We further defer vertices into the set $D$ to ensure that, for the remaining subsets $(T_1, \ldots, T_x, B_1, \ldots, B_\ell)$, the parameter $\eta$ is sufficiently small. To achieve this, we consider collections $\cal E_1, \cal E_2, \ldots, \cal E_y$ of the subsets $B_j$; recall that integer $y = \left \lceil \frac{1 + \vare}{\vare}\ln{\frac1{\vare}} + 1 \right \rceil$ (see equation~\eqref{eq:defnxy}). Here, each collection $\cal E_h = \{B_h, B_{y + h}, B_{2y + h}, \ldots\}$ contains all the subsets whose indices are congruent to $h~\mod y$. 

Fix a collection $\cal E_h$. We define an iterative procedure over the subsets in $\cal E_h$ that defers vertices from these subsets  into the set $D$. Specifically, we iterate over the subsets in $\cal E_h$ in increasing order of index and perform the following operation: if we are at subset $B_{ty + h}$ in the current iteration and it is nonempty, we check whether there exist nonempty subsets $B_{t'y + h}$ with $t' > t$ such that 
\begin{align}
    |B_{t'y + h}| \le (1 + \vare)^{t' - t} |B_{ty + h}|
\end{align} 
If so, we defer all vertices from such subsets $B_{t'y  + h}$ to $D$. When this transfer occurs, we say that $B_{t'y + h}$ has been deferred  {\em due to} $B_{ty + h}$. We repeat this procedure for all collections $\cal E_1, \ldots, \cal E_y$. 

The implication of this procedure is that, within each collection,  the nonempty subsets have geometrically increasing sizes. Specifically, the following claim holds after the processing. 
\begin{claim}\label{claim:bucket-size-upper-bound}
Within any collection $\cal E_h$, for any nonempty subset $B_{ty + h}$ we have 
\begin{align*}
    \sum_{s = 0}^t |B_{sy + h}| \le \left(\frac{1 + \vare}{\vare} \right) |B_{ty + h}|.
\end{align*}
\end{claim}

\begin{proof}
Fix a collection $\cal E_h = \{B_h, B_{y + h}, B_{2y + h}, \ldots\}$ and a nonempty subset $B_{ty + h}$ in it. For any subset $B_{sy + h} \in \cal E_h$ with $s < t$, the deferring procedure executed within $\cal E_h$ ensures that $(1 + \vare)^{t-s} |B_{sy + h}| < |B_{ty + h}|$. Therefore, 
\begin{align}
    \sum_{s=0}^t |B_{sy + h}| \leq \sum_{s=0}^t \frac{1}{(1 + \vare)^{t-s}} |B_{ty + h}| \leq \sum_{s' \geq 0} \frac{1}{(1 + \vare)^{s'}} |B_{ty + h}| \leq \left (\frac{1 + \vare}{\vare} \right) |B_{ty + h}|
\end{align}
The claim stands proved.  
\end{proof}
Next, we bound the weight of the vertices deferred while processing of the collections $\cal E_h$. Consider any nonempty subset $B_j$ and let $W_{\text{def}}(B_j)$ denote the total weight of the vertices deferred \emph{due to} $B_j$.\footnote{Since the deferrals are performed in increasing order of subset indices, $B_j$ remains nonempty at the end of the procedure.} 

Let $B_{t'y + j}$ be any subset whose vertices are deferred due to $B_j$. Then $|B_{t'y + j} | \leq (1+\vare)^{t'} |B_j|$. Also, the vertices in $B_{t'y + j}$ have weight at most $\frac1{(1 + \vare)^{j + t'y - 1}}$. These bounds give us
{
\allowdisplaybreaks
\begin{align}
    W_{\text{def}}(B_j) & \leq \sum_{t'\geq 1} \left( 1 + \vare \right)^{t'} |B_j|  \left( \frac1{(1 + \vare)^{j + t'y - 1}} \right) \nonumber \\
    &\le \frac{|B_j|}{(1 + \vare)^{(j - 1)}} \sum_{t' \geq 1} \left (\frac1{(1 + \vare)^{t'(y - 1)}} \right ) \notag \\   
    &\le \frac{|B_j|}{(1 + \vare)^{(j - 1)}} \sum_{t' \geq 1}  \left (\frac1{(1 + \vare)^{t'\left ( \frac{1 + \vare}{\vare} \ln{\frac{1}{\vare}}\right )}} \right ) \tag{see eqn.~\eqref{eq:defnxy}} \\  
    &\le \frac{|B_j|}{(1 + \vare)^{(j - 1)}} \sum_{t' \geq 1} \vare^{t'} \tag{since $1 + \vare \ge e^{\frac{\vare}{1 + \vare}}$} \\
    &\le \frac{|B_j|}{(1 + \vare)^{(j - 1)}} \left ( \frac{\vare}{1-\vare} \right ) \label{eq:discard-upper-bound}
\end{align}
 }
We collect the nonempty subsets from $T_1,\ldots, T_x, B_1, \ldots, B_\ell$ and order them by nondecreasing size to form the partition $\left( \widetilde{V}_1, \ldots, \widetilde{V}_{\widetilde{\ell}} \right)$ of $V \setminus D$. The above-mentioned construction ensures that each $|\widetilde{V}_j|$ is a multiple of $k$. 

We first bound the parameter $\eta = \max_{t \in [\widetilde{\ell}]} \ \frac{1}{ |\widetilde{V}_t|} \sum_{j=1}^t |\widetilde{V}_j|$ for this partition. The bound will give us that the number of colors  $k \ge \frac{4(1 + \varepsilon)^2\Delta}{\varepsilon^2}\left (\ln{\frac1{\varepsilon}} + 4 \right )$  (as specified in the theorem statement) is sufficient to apply \Cref{thm:partition-equitability}. Hence, invoking \Cref{thm:partition-equitability} with the partition $\left( \widetilde{V}_1, \ldots, \widetilde{V}_{\widetilde{\ell}} \right)$, we obtain a coloring $(\widetilde{C}_1, \ldots, \widetilde{C}_k)$ of the graph $G[V \setminus D]$.

To bound $\eta$, consider any subset $\widetilde{V}_{t}$ in the partition  $\left( \widetilde{V}_1, \ldots, \widetilde{V}_{\widetilde{\ell}} \right)$. Our analysis proceeds via the following two exhaustive cases:

\noindent\textbf{Case 1:} $\widetilde{V}_{t} =T_z$ for some $z$. In this case, $|\widetilde{V}_{t}| = k$ --- the smallest possible size  for a nonempty set in the partition. Hence, all the subsets $\widetilde{V}_{1}, \ldots, \widetilde{V}_{t}$ are of size $k$. Moreover, the index $t$ is at most the number of size-$k$ subsets in the partition. Since there are $x$ subsets $T_1, \ldots, T_x$ (each of size $k$) and every collection $\cal E_1, \ldots, \cal E_y$ contains at most one subset of size $k$,\footnote{Recall that the subset sizes increase geometrically in each collection; see \Cref{claim:bucket-size-upper-bound}.} we have that $t \leq x + y$.  
Therefore, in the current case
\begin{align}
\frac{1}{|\widetilde{V}_{t} |}\sum_{j = 1}^{t} |\widetilde{V}_{j}| = t \le x + y \label{ineq:eta-1}
\end{align}

\noindent\textbf{Case 2:} $\widetilde{V}_{t} = B_{j}$ for some $j$. Let $\widetilde{k} \coloneqq |\widetilde{V}_{t}| \geq k$. To bound the total size of the subsets preceding $\widetilde{V}_{t}$ in the partition $\left( \widetilde{V}_1, \ldots, \widetilde{V}_{\widetilde{\ell}} \right)$, we consider each collection $\cal E_h$ individually: Let $B_{\tau y + h} \in \cal E_h$ be the largest nonempty subset in $\cal E_h$ with $|B_{\tau y+h}| \leq \widetilde{k}$. We consider $B_{\tau y+h}$ in the analysis only if it exists, since such a subset may not exist for every collection. 

We bound the total size of the subsets that lie in $\cal E_h$ and precede $\widetilde{V}_{t}$ in the partition (i.e., subsets in $\{\widetilde{V}_{1}, \ldots, \widetilde{V}_{t} \} \cap \cal E_h$) using \Cref{claim:bucket-size-upper-bound}. Specifically, the total size of such subsets is at most 
\begin{align}
    \sum_{s=0}^\tau |B_{sy + h}| \underset{\text{\Cref{claim:bucket-size-upper-bound}}}{\leq} \left(\frac{1 + \vare}{\vare} \right) |B_{\tau y + h}| \leq \left(\frac{1 + \vare}{\vare} \right) \widetilde{k} = \left(1 + \frac{1}{\vare} \right) |\widetilde{V}_{t}| \label{ineq:per-collection}
\end{align}
    Inequality \eqref{ineq:per-collection} implies that the total size of the subsets in $\{\widetilde{V}_{1}, \ldots, \widetilde{V}_{t} \} \cap \left( \cup_{h=1}^y \cal E_h \right)$ is at most $y \left(1 + \frac{1}{\vare} \right) |\widetilde{V}_{t}|$. In addition to these subsets from $\cup_{h=1}^y \cal E_h$, at most $x$ subsets $T_1, \ldots, T_x$ (each of size $k$) precede $\widetilde{V}_{t}$. Therefore, 
    \begin{align}
        \sum_{j=1}^t |\widetilde{V}_j| \leq y \left(1 + \frac{1}{\vare} \right) |\widetilde{V}_{t}| + x k \leq y \left(1 + \frac{1}{\vare} \right) |\widetilde{V}_{t}| + x |\widetilde{V}_{t}| 
    \end{align}
    The inequality reduces to 
    \begin{align}
        \frac{1}{|\widetilde{V}_{t}| } \sum_{j=1}^t |\widetilde{V}_j| \leq y \left(1 + \frac{1}{\vare} \right) + x  \label{ineq:eta-2}
    \end{align}

Considering the two cases (inequalities \eqref{ineq:eta-1} and \eqref{ineq:eta-2}), we obtain 
\begin{align}
    \eta  = \max_{t \in [\widetilde{\ell}]} \ \frac{1}{ |\widetilde{V}_t|} \sum_{j=1}^t |\widetilde{V}_j| 
     & \leq y \left(1 + \frac{1}{\vare} \right) + x \nonumber \\
    & \leq \left ( \frac{1 + \vare}{\vare} \ln{\frac1{\vare}} + 2 \right ) \left (1 + \frac1\vare \right ) + \frac{1}{\vare^2} + 1 \tag{eqn.~\eqref{eq:defnxy}} \\
&= \left ( \frac{1+ \vare}{\vare} \right )^2 \ln{\frac{1}{\vare}} + \left ( \frac{1+ \vare}{\vare} \right )^2 + 2. \notag \\
&\le \left ( \frac{1+ \vare}{\vare} \right )^2 \left (\ln{\frac1{\vare}} + 2 \right ) \label{ineq:ub-eta}
\end{align}
This upper bound on $\eta$ implies that $k$ (as specified in the theorem statement) is large enough to invoke \Cref{thm:partition-equitability} on the graph $G[V \setminus D]$ with partition $\left( \widetilde{V}_1, \ldots, \widetilde{V}_{\widetilde{\ell}} \right)$, yielding a $k$-coloring  $\left( \widetilde{C}_1, \ldots, \widetilde{C}_k \right)$ that is equitable with respect to this partition:
\begin{align}
    |\widetilde{C}_i \cap \widetilde{V}_j| = \frac{|\widetilde{V}_j|}{k} \quad \text{ for every } i \in [k] \text { and } j \in [\widetilde{\ell}] \label{eq:tildeCeq}
\end{align}

Using Lemma \ref{lem:filling-up} to color the deferred vertices $D$, we extend $\left( \widetilde{C}_1, \ldots, \widetilde{C}_k \right)$ to obtain a $k$-coloring $(C^*_1, \ldots, C^*_k)$ of $G$. 

We complete the proof of the theorem by showing that $(C^*_1, \ldots, C^*_k)$ is $(1 + 3\vare)$-\EQo under the given weights. The following three claims will be used to establish this equitability guarantee.

\begin{claim}\label{claim:discard-multiple-attribution}
For every color class $C^*_i$ in the coloring $(C^*_1, \ldots, C^*_k)$, we have $\sum_{j = 1}^x w(C^*_i \cap T_j) \ge x \ge \frac1{\vare^2}$. 
\end{claim}
\begin{proof}
The coloring $(C^*_1, \ldots, C^*_k)$ is obtained by extending $\left( \widetilde{C}_1, \ldots, \widetilde{C}_k \right)$ via Lemma \ref{lem:filling-up}. Hence, $\widetilde{C}_i \subseteq C^*_i$ for each $i \in [k]$. Furthermore, equitability of $\left( \widetilde{C}_1, \ldots, \widetilde{C}_k \right)$ with respect to the partition $\left( \widetilde{V}_1, \ldots, \widetilde{V}_{\widetilde{\ell}} \right)$ (see equation \eqref{eq:tildeCeq}) ensures that $|\widetilde{C}_i \cap T_j| = |T_j|/k = 1$ for each $i \in [k]$ and $j \in [x]$; recall that the subsets $T_j$ are parts in the partition $\left( \widetilde{V}_1, \ldots, \widetilde{V}_{\widetilde{\ell}} \right)$. Summing over $j \in [x]$, we obtain $\sum_{j = 1}^x w(C^*_i \cap T_j) = \sum_{j = 1}^x w \left(\widetilde{C}_i \cap T_j \right) \ge x$. The last inequality follows from the fact that, under the previously-mentioned normalization, each vertex in $T_j$ has weight at least $1$. The claim stands proved. 
\end{proof}

\begin{claim}\label{claim:discarded-weight-upper-bound}
The total weight of the deferred set $D$ satisfies $w(D) \leq k\vare \left(\frac{1 + \vare}{1 - \vare} \right) w(C^*_i)$ for each color class $C^*_i$.
\end{claim}
\begin{proof}
The weight of the set $D$ can be upper bounded by combining inequalities \eqref{eq:discard-multiple} and \eqref{eq:discard-upper-bound}:
{\allowdisplaybreaks
\begin{align*}
w(D) &= W'_{\text{def}} + \sum_{j \in [\ell]} W_{\text{def}}(B_j) \\
&\le \frac{k(1 + \vare)}{\vare} \ + \ \sum_{j \in [\ell]} \left ( \frac{\vare}{1-\vare} \right )\left (\frac{|B_j|}{(1 + \vare)^{(j - 1)}} \right ) \\
&\le k(1 + \vare)\vare \sum_{j = 1}^x w(C^*_i \cap T_j) \ +  \ \sum_{j \in [\ell]} \left ( \frac{\vare}{1-\vare} \right )\left (\frac{|B_j|}{(1 + \vare)^{(j - 1)}} \right ) \tag{\Cref{claim:discard-multiple-attribution}} \\
&\le k(1 + \vare)\vare \sum_{j = 1}^x w(C^*_i \cap T_j) \ + \ \sum_{j \in [\ell]} \left ( \frac{\vare}{1-\vare} \right ) \ k(1 + \vare) \ w(C^*_i \cap B_j)\\
&\le k \vare \left( \frac{ 1 + \vare}{1-\vare} \right) \ w(C^*_i).
\end{align*}
}
The penultimate inequality follows from the fact that each $\widetilde{C}_i$ contains exactly $|B_j|/k$ vertices from the subset $B_j$ (see equation \eqref{eq:tildeCeq}) and each of these vertices has weight at least $\frac1{(1 + \vare)^j}$. This completes the proof of the claim. 
\end{proof}

\begin{claim}\label{claim:math-stuff-1}
$\left (\frac{k}{k-\Delta} \right ) \left ( \frac{1 + \vare}{1 - \vare} \right ) \le 1.5$.
\end{claim}
\begin{proof}
Since $\vare \le 0.1$, we have $k \geq 400\Delta$ and $\left ( \frac{1 + \vare}{1 - \vare} \right ) \le \frac{11}{9}$. Hence, 
\begin{align*}
    \left (\frac{k}{k-\Delta} \right ) \left ( \frac{1 + \vare}{1 - \vare} \right ) \le \frac{400}{399} \cdot \frac{11}{9} \le 1.5.
\end{align*}
\end{proof}

We now complete the proof of the theorem by establishing that the coloring $(C^*_1,\ldots,C^*_k)$ satisfies the stated approximate \EQo guarantee under the given weights. Consider any two color classes $C^*_a$ and $C^*_{b}$. Let $\bar{v}$ be the vertex in the $C^*_a \cap T_1$.
{\allowdisplaybreaks
\begin{align}
w(C^*_a \setminus \{\bar{v}\}) & \le w\left( \widetilde{C}_a \setminus \{\bar v\} \right) + \frac{w(D)}{k - \Delta} + \maxw(D) \tag{via \Cref{lem:filling-up}} \\
& \le w\left( \widetilde{C}_a \setminus \{\bar v\} \right) + \frac{w(D)}{k - \Delta} + 1 \tag{$w_v \leq 1$ for every $v\in D$, under normalization} \\ 
&\le \sum_{j = 2}^x \maxw(T_j) + \sum_{j \in [\ell]} \frac{|B_j|}{k} \maxw(B_j) + \frac{w(D)}{k - \Delta} + 1  \tag{via equation \eqref{eq:tildeCeq}}\\
&\le \sum_{j = 1}^{x-1} \minw(T_j) + \sum_{j \in [\ell]} \frac{|B_j|}{k} \maxw(B_j) + \frac{w(D)}{k - \Delta} + 1 \tag{$\minw(T_j) \ge \maxw(T_{j+1})$} \\
&\le \sum_{j = 1}^{x-1} \minw(T_j) + \sum_{j \in [\ell]} \frac{|B_j|}{k} (1 + \vare) \minw(B_j) + \frac{w(D)}{k - \Delta} + 1 \tag{$w_v \in \left(\frac{1}{(1 + \vare)^{j}}, \frac{1}{(1 + \vare)^{j-1}} \right]$ for all $v\in B_j$} \\
&\le \sum_{j = 1}^{x-1} w(C^*_{b} \cap T_j) + (1 + \vare) \sum_{j \in [\ell]} w(C^*_{b} \cap B_j) \ \ + \frac{w(D)}{k - \Delta} + 1  \tag{equation \eqref{eq:tildeCeq} and $\widetilde{C}_{b} \subseteq C^*_{b}$} \\
&\le (1 + \vare) w(C^*_{b}) \ \ + \frac{w(D)}{k - \Delta} + 1  \label{ineq:apx-eqo}
\end{align}
Furthermore, note that $\vare^2 w(C^*_{b}) \geq 1$ (\Cref{claim:discard-multiple-attribution}) and $w(D) \leq k\vare \left(\frac{1 + \vare}{1 - \vare} \right) w(C^*_{b})$ (\Cref{claim:discarded-weight-upper-bound}). Hence, inequality~\eqref{ineq:apx-eqo} extends to

\begin{align*}
w(C^*_a \setminus \{\bar{v}\}) & \le (1 + \vare) w(C^*_{b})  + \frac{k \vare}{k-\Delta} \left(\frac{1 + \vare}{1 - \vare} \right) w(C^*_{b}) + \vare^2 w(C^*_{b}) \\  
& \le (1 + \vare) w(C^*_{b})  \ + \ 1.5 \vare  \  w(C^*_{b}) \ + \ 0.1 \vare \ w(C^*_{b}) \tag{\Cref{claim:math-stuff-1} and $\vare \leq 0.1$} \\
&\le (1 + 3 \vare) w(C^*_{b}).
\end{align*}
}
Overall, we obtain that the coloring $(C^*_1,\ldots,C^*_k)$ is $(1+3\vare)$-\EQo. 

\subsection{Efficient Algorithm}
The proof of \Cref{thm:eps-equitable} yields a polynomial-time algorithm (Algorithm \ref{algo:eps-eqo}) to compute a $(1 + 3\vare)$-\EQo coloring.   

\begin{algorithm}[h]
\caption{Finding a $(1 + 3\vare)$-\EQo Coloring}
\label{algo:eps-eqo}
\begin{algorithmic}[1]
\Require Vertex-weighted graph $G=(V, E, w)$ and integer $k \in \mathbb{Z}_+$.
\Ensure A $(1 + 3\vare)$-\EQo $k$-coloring $(C^*_1, \dots, C^*_k)$. 
\State Initialize $D = \emptyset$
\State Partition the vertices into $(T_1, \dots, T_x, B_1, \dots, B_{\ell})$ as defined in Section \ref{sec:apx-eqo-proof}
\State Defer vertices in each bucket to $D$ to ensure each $|B_i|$ is a multiple of $k$ \label{line:start}
\State Create collections $\cal E_1, \dots, \cal E_{y}$ as defined in Section \ref{sec:apx-eqo-proof}
\For{each collection $\cal E_j$}
    \For{each non-empty bucket $B_{ty + j} \in \cal E_j$ in increasing order of index $t$}
        \If{there exists a non-empty bucket $B_{t'y + j}$ such that $t' > t$ and $|B_{t'y + j}| \le (1 + \vare)^{t' - t}|B_{ty + j}|$}
        \State Defer all vertices in the bucket $B_{t'y + j}$ to $D$
        \EndIf
    \EndFor
\EndFor
\State Let $(\widetilde{V}_1, \dots, \widetilde{V}_{\ell})$ be the non-empty sets of $(T_1, \dots, T_x, B_1, \dots, B_{\ell})$ sorted in increasng order of size
\State Use Theorem \ref{thm:partition-equitability} to construct a coloring $(\widetilde{C}_1, \dots, \widetilde{C}_k)$ of $G[V \setminus D]$
\State Complete the coloring $(\widetilde{C}_1, \ldots, \widetilde{C}_k)$ using Lemma~\ref{lem:filling-up} to obtain a $k$-coloring $(C^*_1, \dots, C^*_k)$ of $G$. \label{line:end}
\State \Return $(C^*_1, \dots, C^*_k)$.
\end{algorithmic}
\end{algorithm}

The only issue preventing a trivial implementation of the above algorithm to run in polynomial time is that the number of buckets $\ell$ could be exponential in $n$. To sidestep this, we use the fact that there are at most $n$ non-empty buckets. Moreover, we can identify the set in $T_1, \dots, T_x, B_1, \dots, B_{\ell}$ that each vertex $i$ belongs to in polynomial time. This is because for any rational number $p/q$, we can find the $j$ such that $\left ( \frac{1}{(1 + \vare)^j}, \frac{1}{(1 + \vare)^{j-1}}\right ]$ in time polynomial in $\frac1{\vare}$ and $O(\max\{\log p, \log q\})$. Note that the number of bits required to represent the rational number $p/q$ is $O(\max\{\log p, \log q\})$. Therefore, in polynomial time, we can identify which buckets are non-empty and what each non-empty bucket contains. Using this, all the remaining steps of the algorithm (Lines \ref{line:start}--\ref{line:end}) can be carried out in polynomial time.

\begin{remark}
The algorithm (and the proof) implicitly assumes that no vertex in $V$ has zero weight; there is no bucket that these vertices can be added to. Zero weight vertices can be accommodated as follows: let $V_0$ denote the set of all zero weight vertices. We invoke Theorem \ref{thm:eps-equitable} with the graph $G[V \setminus V_0]$, to get a coloring $(C_1, \dots, C_k)$ for the graph $G[V \setminus V_0]$. Then, we complete this coloring using Lemma \ref{lem:filling-up} to obtain a coloring $(C^*_1, \dots, C^*_k)$ for the graph $G$. This completion step only adds zero weight vertices to each color class and therefore, the final coloring remains $(1 + 3\vare)$-\EQo.
\end{remark}

\section{$(1+\vare)$-\EQo with Low Maximum Weight}\label{sec:low-max-wt-eqo}
Our next set of results shows that when the maximum weight $\max_{v \in V} w_v$ is small relative to the total weight $w(V)$, we can obtain stronger guarantees. 

\begin{theorem}\label{thm:low-max-weight}
Let $G = (V, E, w)$ be a vertex-weighted graph with maximum degree $\Delta$, and let $(C'_1, \dots, C'_\kappa)$ be any $\kappa$-coloring of the graph $G$. For any parameter $\vare \in (0, 0.1]$ and integer $k \in \mathbb{Z}_+$, if $k \ge \max \left \{ \frac{\kappa}{\vare}, 2\Delta \right \}$ and $\maxw(V) \le \vare \frac{w(V)}{k}$, then there exists a $(1 + 7\vare)$-\EQo $k$-coloring of $G$. Furthermore, given the $\kappa$-coloring $(C'_1,\dots, C'_\kappa)$, such a coloring can be computed in time polynomial in $n$ and $\frac1{\vare}$.
\end{theorem}
\begin{proof}
Normalize the weights such that $w(V) = 1$. We partition every color classes $C'_i$ of $(C'_1, \dots, C'_\kappa)$ into $\ell_i \in \mathbb{Z}_+$ pairwise disjoint subsets $D_{(i, 1)}, \ldots D_{(i, \ell_i)}$ such that, for each $j \in [\ell_i]$:
\begin{enumerate}[(a)]
    \item $D_{(i, j)} \subseteq C'_i$.
    \item $w(D_{(i, j)}) \in [\frac{1 - 2\vare}{k}, \frac{1 - \vare}k]$.
    \item There exists some $\bar{v} \in D_{(i, j)}$ such that $w(D_{(i, j)} \setminus \{\bar v\}) \le \frac{1 - 2\vare}{k}$
    \item $w(C'_i \setminus \bigcup_{j = 1}^{\ell_i} D_{(i, j)}) \le \frac{1}{k}$.
\end{enumerate}
Since $\maxw(C'_i) \le \frac{\vare}{k}$,\footnote{Recall that $w(V) = 1$.} the subsets $D_{(i, 1)}, \ldots D_{(i, \ell_i)}$ satisfying properties (a)-(d) exist and can be constructed via a greedy procedure. Specifically, for each $j$, form $D_{(i, j)}$ as any minimal subset of $C'_i \setminus (\bigcup_{j' = 1}^{j-1} D_{(i, j')})$ with weight at least $\frac{1-2\vare}{k}$, provided such a subset exits. 

Now, consider vertex subset $V' \coloneqq \cup_{i=1}^{\kappa} \cup_{j=1}^{\ell_i} D_{(i, j)}$ and note that $\left(D_{(1, 1)}, \ldots, D_{(1, \ell_1)}, \ldots, D_{(\kappa,1)}, \ldots D_{(\kappa, \ell_\kappa)} \right)$ forms a coloring for the induced subgraph $G[V']$. The total weight of these color classes satisfies $w(V') \geq 1 - \frac{\kappa}{k} \ge 1 - \vare$; this follows from property (d) and the bound on $k$ in the theorem statement.

Re-index these color classes as $(D_1, \dots, D_{k'})$. Since each $D_i$ has weight at most $\frac{1 - \vare}{k}$ and the total weight of these color classes $w(V')$ is at least $1 - \vare$, it follows that $k' \ge k$. Now, consider $k$ color classes $D_1,\ldots, D_k$ and note that $(D_1, \ldots, D_k)$ forms a $k$-coloring of the subgraph $G[V'']$, where $V'' = \cup_{i=1}^k D_i \subseteq V'$. Since each $D_i$ has weight $w(D_i) \geq \frac{1 - 2\vare}{k}$, we have $w(V \setminus V'') \le 2\vare$.

We next apply Lemma \ref{lem:filling-up} to extend $(D_1, \ldots, D_k)$ and obtain a $k$-coloring $(D^*_1, \ldots, D^*_k)$ of the given graph $G$. We will show that this coloring is $(1 + 7\vare)$-\EQo with respect to the vertex weights. To this end, fix any two color classes $D^*_i$ and $D^*_j$, and let vertex $\bar{v} \in D_i \subseteq D^*_i$ be such that $w(D_i \setminus \{\bar v\}) \le \frac{1 - 2\vare}{k}$; the existence of such a vertex follows from property (c). We have
\begin{align*}
    w(D^*_i \setminus \{\bar v\}) &= w(D_i \setminus \{\bar v\}) + w(D^*_i \setminus D_i) \\
    & \leq \frac{1 - 2\vare}{k} + w(D^*_i \setminus D_i) \tag{property (c)}\\
    &\le \frac{1 - 2\vare}{k} + \frac{w(V \setminus V'')}{k - \Delta} + \maxw(V \setminus V'') \tag{Lemma \ref{lem:filling-up}}\\
    &\le  \frac{1 - 2\vare}{k} + \frac{2\vare}{k - \Delta} + \frac{\vare}{k} \\ 
    &\le \frac{1 - 2\vare}{k} + \frac{4\vare}{k} + \frac{\vare}{k} \tag{since $k \geq 2\Delta$} \\
    &= \frac{1 - 2\vare}{k} \left ( 1 + \frac{5 \vare}{1 - 2 \vare} \right ) \\
    &\le w(D_j) (1 + 7\vare) \tag{since $\vare \leq 0.1$ and $w(D_j) \geq \frac{1 - 2\vare}{k}$} \\ 
    & \le w(D^*_j)(1 + 7\vare).
\end{align*}
Therefore, the $k$-coloring $(D^*_1, \ldots, D^*_k)$ of $G$ is $(1 + 7\vare)$-\EQo.

The proof yields a polynomial-time algorithm for computing the desired $(1 + 7\vare)$-\EQo coloring. First, for each given color class $C'_i$, we construct a sub-partition $\left(D_{(i, 1)}, \dots, D_{(i, \ell_i)} \right)$ satisfying properties (a)--(d) using the efficient greedy procedure described above. Once the sub-partition $(D_{(i, 1)}, \dots, D_{(i, \ell_i)})$ is constructed for all color classes $C'_i$, the other two steps of the algorithm are straightforward: 
\begin{inparaenum}[(i)]
    \item select any $k$ sets $D_1, \dots, D_{k}$, and
    \item invoke Lemma \ref{lem:filling-up} to complete the coloring.
\end{inparaenum}
Both steps can be executed in polynomial time. This completes the proof. 
\end{proof}

We can instantiate \Cref{thm:low-max-weight} with different colorings $(C'_1, \ldots, C'_\kappa)$ to derive the following two corollaries. First, note that every graph $G$ admits a $(\Delta + 1)$-coloring $(C'_1, \dots, C'_{\Delta + 1})$, and such a coloring can be computed in polynomial time~\cite{Brooks1941Colouring}. Hence, we can invoke 
Theorem \ref{thm:low-max-weight} with a coloring $(C'_1, \dots, C'_{\Delta + 1})$ and obtain the following corollary. 

\thmlowmaxwtcomputable*

The existential guarantee in the above corollary can be strengthened by applying Theorem \ref{thm:low-max-weight} with a $\chi(G)$-coloring of the graph $G$.\footnote{Recall that $\chi(G)$ denotes the chromatic number of the graph $G$.} However, since a $\chi(G)$-coloring need not be computable in polynomial time, the following corollary does not come with an efficient algorithm.

\thmchromaticnumber*

\section{Conclusion and Future Work}
This work establishes bounds on the number of colors $k$ required to obtain approximately equitable colorings in vertex-weighted graphs. In particular, we prove that equitability up to one vertex and within a $(1+\varepsilon)$ multiplicative factor---i.e., a $(1+\varepsilon)$-\EQo coloring---can be achieved using $k=\widetilde{O} \left(\Delta/\varepsilon^2\right)$ colors. In addition, we show that the      Hajnal-Szemer\'{e}di threshold $k \ge \Delta+1$ guarantees $2$-\EQo in the weighted setting, and a multiplicative approximation is unavoidable at this threshold. A notable aspect of these results is that they hold without any assumptions on the range of the vertex weights.

In addition, we develop results for colorings that are equitable with respect to arbitrary partitions of the vertex set. These partition-equitability guarantees provide a novel approach to balancing weights and may be of independent interest.

Complementing the $(1+\varepsilon)$-\EQo coloring result, it would be interesting to establish a lower bound demonstrating that some dependence of $k$ on $\varepsilon$ is unavoidable. In parallel, an appealing direction for future work is to determine whether there exists a fixed constant $c$ such that an $\EQo$ $k$-coloring always exists for every $k \ge c\Delta$.

\section*{Acknowledgments}
Siddharth Barman gratefully acknowledges the support of the Walmart Centre for Tech Excellence (CSR WMGT-23-0001) an Ittiam CSR Grant (OD/OTHR-24-0032).
Vignesh Viswanathan is funded by the National Science Foundation (NSF) Career Award 2441296 and Grant RI-2327057. A part of this research was completed while Vignesh was a visitor at the Indian Institute of Science, Bangalore.
\addcontentsline{toc}{section}{References}
\bibliographystyle{alpha}
\bibliography{abb,references}

\appendix 

\section{Randomized Algorithm for Equitable Coloring}\label{apdx:probabilistic-method}
In this section we use the probabilistic method to prove the existence of an approximate \EQo coloring. Our analysis is similar to that of \cite{Pemmaraju2008SymmetryBreaking}. Assume via normalization that $\maxw(V) = 1$ and also $k \ge \Delta + 1$. \\

\noindent\textbf{Random Coloring Generation Process: } Pick a uniformly random permutation $\pi$ of the vertices $V$, and assign to each vertex a tentative color chosen uniformly at random from the set of colors $[k]$. For each vertex $v$, let $\mathcal{E}_v$ denote the event that there exists a neighbor $v'$ of $v$ that appears before $v$ in $\pi$ and is assigned the same tentative color.

If $\mathcal{E}_v$ does not occur, assign $v$ its tentative color as a permanent color; otherwise, leave $v$ uncolored. The final output is the (partial) coloring consisting of the permanent colors assigned to the vertices. Let $(\mathbf{C}_1, \ldots, \mathbf{C}_k)$ denote a partial coloring generated by this process. Also, let $\Pr \{ v \in \mathbf{C}_i \}$ denote the probability that vertex $v$ is assigned color class $i$. 

The proof of the following lemma is similar that of Lemma 5 in~\cite{Pemmaraju2008SymmetryBreaking}; we provide it here for completeness.

\begin{lemma}\label{lem:expectation}
$\Pr\left\{ v \in \mathbf{C}_i \right\} \in \left [\frac1{k+ \Delta + 1}, \frac1k \right]$ for all vertices $v \in V$ and colors $i \in [k]$. 
\end{lemma}
\begin{proof}
Write ${\rm deg}(v)$ to denote the degree of vertex $v$ and let $L(v)$ denote the subset of neighbors of $v$ that appear before $v$ in the permutation $\pi$. By summing over all possible sizes of $L(v)$, we get
\begin{align}
    \Pr \{ v \in C_i \} &= \frac1k \sum_{j = 0}^{{\rm deg}(v)} \frac1{{\rm deg}(v) + 1} \left (1 - \frac1k \right )^j \notag\\
    &= \frac1{ {\rm deg}(v) + 1} \left ( 1 - \left ( 1- \frac1k \right )^{ {\rm deg}(v) + 1}\right ). \label{eq:probability}
\end{align}
The right hand side of the first equality has a $\frac1k$ term to denote the probability that $v$ is tentatively assigned the color $i$. The summation is over all possible sizes of $L(v)$, with each size occurring with probability $\frac1{{\rm deg}(v) + 1}$. 

The right hand side of \eqref{eq:probability} is maximized at ${\rm deg}(v) = 0$ and minimized at ${\rm deg}(v) = \Delta$. Therefore $\Pr\{v \in C_i\} \le \frac1k$ and
\begin{align*}
    \Pr\{v \in C_i\} \ge \frac1{\Delta + 1} \left ( 1 - \left ( 1- \frac1k \right )^{\Delta + 1}\right ) \ge \frac1{\Delta + 1} \left ( 1 - \frac1{1 + \frac{\Delta + 1}{k}}\right ) = \frac1{k + \Delta + 1}.
\end{align*}
The second inequality above follows from the identity $\left (1-x \right )^{r} \le \frac{1}{1 + rx}$, which holds for $x \in [0, 1]$ and $r \ge 0$.
\end{proof}

We use the Azuma-Hoeffding inequality \cite{Azuma1967} to bound how much the weight of each color class deviates from its expectation.

\begin{theorem}[Azuma's Inequality]\label{thm:azuma}
Let $X$ be a random variable determined by $n$ trials $T_1, T_2, \ldots, T_n$, such that for each $i$, and any two possible sequences of outcomes
$t_1, t_2, \ldots, t_{i-1}, t_i$ and $t_1, t_2, \ldots, t_{i-1}, t_i'$:
\begin{align*}
    \left|\E[X \mid T_1 = t_1, \ldots, T_i = t_i] - \E[X \mid T_1 = t_1, \ldots, T_i = t_i']
    \right| \leq c_i
\end{align*}
then
\begin{align*}
       \Pr \left\{ |X - E[X]| > t \right\}
    \leq 2 {\rm exp} \left( -\frac{t^2}{2 \sum_{i=1}^n c_i^2} \right). 
\end{align*}
\end{theorem}

The following application of the Azuma-Hoeffding inequality follows the arguments of \cite{Pemmaraju2008SymmetryBreaking}; we include it here for completeness.

\begin{lemma}\label{lem:tail-bound}
For all $i \in [k]$, 
\begin{align*}
\Pr \left\{ \left |w(\mathbf{C}_i) - \E[w(\mathbf{C}_i)]\right | > c \sqrt{w(V) \ln{k}} \right\} \le \frac{2}{k^{c^2/18}}.
\end{align*}
\end{lemma}
\begin{proof}
Fix a permutation $\pi$. Let $T_j$ be the tentative color assignment of the $j$-th vertex in $\pi$ and note that $w(\mathbf{C}_i)$ is completely determined by the outcomes of the trials $T_1, T_2, \ldots, T_n$.
Now consider the difference in conditional expectations from the Azuma-Hoeffding inequality:
\begin{align}
    c_j = \left|
        \E[w(\mathbf{C}_i) \mid T_1 = t_1, T_2 = t_2, \ldots, T_j = t_j]
        - \E[w(\mathbf{C}_i) \mid T_1 = t_1, T_2 = t_2, \ldots, T_j = t_j'] \label{eq:azuma-diff}
    \right|.
\end{align}
Let $v$ be the $j$-th vertex in $\pi$. The difference in \Cref{eq:azuma-diff} is at most
\[
    c_j \leq w_v + \frac{2}{k} \sum_{u \in U(v)} w_u,
\]
where $U(v)$ is the set of neighbors of $v$ that appear after $v$ in $\pi$.
The first term $w_v$ in the above bound is due to the change in the tentative
color of $v$ from $t_j$ to $t_j'$. The second term is the \textit{expected} change
in $w(\mathbf{C}_i)$; note that this occurs only at vertices in $U(v)$. Using the fact that $\maxw(V) = 1$, we obtain $c_j \leq 1 + \frac{2\Delta}{k} < 3$. 
Also,
\begin{align*}
\sum_{j = 1}^{n} c_j
\leq \sum_{v \in V} w_v + \frac{2}{k} \sum_{v \in V} \sum_{u \in U(v)} w(u)
\leq w(V) + \frac{2\Delta}{k} \  \sum_{v \in V} w(v)
< 3w(V).  
\end{align*}

\noindent These inequalities imply $\sum_{j=1}^{n} c_j^2 < 3\sum_{j=1}^{n} c_j < 9w(V)$. With this bound, we apply the Azuma–Hoeffding inequality (\Cref{lem:tail-bound}) to obtain the stated upper bound on the deviation from the mean: 
\begin{align*}
\Pr \left \{\left |w(\mathbf{C}_i) - \E[w(\mathbf{C}_i)]\right | > c \sqrt{w(V) \ln{k}} \right\} \le 2\exp\!\left(\frac{-c^2 w(V) \ln{k}}{18w(V)}\right) = \frac{2}{k^{c^2/18}}.
\end{align*}
\end{proof}

We next establish the main theorem of this section.
\thmprobabilisticmethod*
\begin{proof}
 We generate a random coloring $(\mathbf{C}_1, \ldots, \mathbf{C}_k)$ using the random coloring generation process described above. Applying the linearity of expectation with Lemma \ref{lem:expectation}, we obtain $\E[w(\mathbf{C}_i)] \in \left [\frac{w(V)}{k + \Delta + 1}, \frac{w(V)}{k} \right ]$. Hence, the assumption that that $k \ge \frac{\Delta + 1}{\vare}$, gives us 
\begin{align*}
\E[w(\mathbf{C}_i)] \in \left [\frac{w(V)}{(1 + \vare)k}, \frac{w(V)}{k} \right ].
\end{align*} 
Using $c = 6$ in Lemma \ref{lem:tail-bound}, we get, for all $i \in [k]$,
\begin{align}
\Pr \left\{\left |w(\mathbf{C}_i) - \E[w(\mathbf{C}_i)]\right | > 6 \sqrt{w(V) \ln{k}} \right\} \le \frac{2}{k^{2}}  \label{eq:tail-bound-inequality}
\end{align}
Now, the union bound gives us
\begin{align*}
\Pr \left\{ \text{ there exists } i \in [k] \text{ s.t.} \left |w(\mathbf{C}_i) - \E[w(\mathbf{C}_i)]\right | > 6 \sqrt{w(V) \ln{k}} \right\} \le \frac{2}{k} < 1.
\end{align*}

Therefore, a (partial) coloring where the event in \eqref{eq:tail-bound-inequality} does not occur for all $i \in [k]$ exists. Denote this coloring as $(C_1, \dots, C_k)$. By our assumption in the theorem, $\sqrt{w(V) \ln{k}} \le \frac{\vare w(V)}{k}$; recall that via normalization we have $\maxw(V) = 1$. Therefore, for all $i \in [k]$,
\begin{align}
    w(C_i) \in \left [ \frac{w(V)(1 - \vare)}{k} - \frac{6\vare w(V)}{k}, \frac{w(V)}{k} + \frac{6\vare w(V)}{k}\right]. \label{eq:weight-bounds}
\end{align}

The coloring $(C_1, \dots, C_k)$ is a coloring for some induced subgraph $G[V']$. From the above bound, we get that the total weight of the uncolored vertices $w(V \setminus V')$ is at most $6\vare w(V)$.
We complete this coloring to create a coloring $(C'_1, \dots, C'_k)$ of the graph $G$ using Lemma \ref{lem:filling-up}. We complete the proof next by showing that this coloring is $(1 + 25\vare)$-\EQo. Consider any two color classes $C'_i$ and $C'_j$. Let $\bar{v}$ denote the highest weight vertex in $C'_i \setminus C_i$. Using Lemma \ref{lem:filling-up},
\begin{align*}
    w(C'_i \setminus \bar{v}) &\le w(C_i) + \frac{w(V \setminus V')}{k-\Delta} \\
    &\le \frac{(1 + 6\vare)w(V)}{k} + \frac{6\vare w(V)}{k - \Delta} \\
    &\le \frac{(1 + 6\vare)w(V)}{k} + \frac{6\vare w(V)}{k - \vare k} \\
    &\le \frac{w(V)}{k} \left ( 1 + 6\vare + 7\vare \right ) \\
    &\le w(C_j)\frac{1 + 13\vare}{1 - 7\vare} \tag{via (\eqref{eq:weight-bounds})} \\
    &\le w(C'_j)(1 + 13\vare)(1 + 10\vare) \le w(C'_j)(1 + 25\vare)
\end{align*}
In both the fourth and sixth inequalities, we use $\vare < 0.01$ to simplify the expression. The theorem stands proved. 
\end{proof}

\section{Concentration Inequalities}
In this section, we prove a concentration inequality for sums of partly dependent random variables using our equitable coloring results.

\thmtailbounds*
\begin{proof}
For each vertex $i$ in the graph $G$, we define its weight $w_i$ to be $(b_i - a_i)^2$. We invoke Theorem \ref{thm:delta-plus-one} to construct a $2$-\EQo coloring $(C_1, \dots, C_{k})$ with $k = \Delta + 1$. 

Each color class $C_j$ is an independent set in $G$, which means the random variables corresponding to the vertices of $C_j$ are independent. Let $X^j$ denote the sum of the random variables in $C_j$; that is, $X^j = \sum_{i \in C_j} X_i$. 

For each $j \in [\Delta + 1]$, we can apply Hoeffding's inequality \cite{Hoeffding1963Inequality}. This gives us,
\begin{align}
    \Pr \left\{X^j - \E[X^j] \ge \frac{t}{\Delta + 1} \right\} \le \exp \left ( \frac{-2t^2}{(\Delta + 1)^2 \left [ \sum_{i \in C_j} (b_i - a_i)^2 \right ] }\right ) \label{eq:hoeffding}
\end{align}

Note that the term $\sum_{i \in C_j} (b_i - a_i)^2$ is equal to the weight of the set $C_j$, $w(C_j)$, in the graph $G$. We use the definition of $2$-\EQo to upper bound $w(C_j)$. 

Specifically, let $\overline{v}$ denote the highest weight vertex in $C_j$. By the definition of a $2$-\EQo coloring, we have for any $j' \in [\Delta + 1]$, 
\begin{align*}
w(C_j) \le 2w(C_{j'}) + w_{\overline{v}} \le 2w(C_{j'}) + \maxw([n]). 
\end{align*}

Summing this inequality up for all $j' \in [\Delta + 1]$, we get
\begin{align}
    (\Delta + 1)w(C_j) \le 2w([n]) + (\Delta + 1)\maxw([n]). \label{eq:cj-upper-bound}
\end{align}

Applying the upper bound from Equation \eqref{eq:cj-upper-bound} into Equation \eqref{eq:hoeffding}, we get for each $j \in [\Delta + 1]$
\begin{align}
\Pr \left\{X^j - \E[X^j] \ge \frac{t}{\Delta + 1} \right\} \le \exp \left ( \frac{-2t^2}{2(\Delta + 1)\sum_{i \in [n]} (b_i - a_i)^2 + (\Delta + 1)^2\max_{i \in [n]} (b_i - a_i)^2}\right ).\label{eq:complete-hoeffding}   
\end{align}

The theorem then follows by taking the union bound over all $j \in [\Delta + 1]$. 
\begin{align*}
    \Pr \{X - \E[X] \ge t \} &\le \Pr \left\{ \text{there exists } j \in [k]  \text{ s.t. } X^j - \E[X^j] \ge \frac{t}{\Delta + 1} \right \} \\
    &\le \sum_{j \in [\Delta + 1]} \Pr \left\{ X^j - \E[X^j] \ge \frac{t}{\Delta + 1} \right\} \tag{Union bound}\\
    &\le (\Delta + 1) \exp \left( \frac{-2t^2}{2(\Delta + 1)\sum_{i \in [n]} (b_i - a_i)^2 + (\Delta + 1)^2\max_{i \in [n]} (b_i - a_i)^2 }\right) \tag{via \eqref{eq:complete-hoeffding}}
\end{align*}

The upper bound for the probability $\Pr\{X - \E[X] \le -t\}$ follows from a similar argument. This completes the proof of the theorem. 
\end{proof}

\section{On the (Non-)Existence of EQX colorings}\label{apdx:eqx}
This section studies a strengthening of $\alpha$-\EQo, called $\alpha$-approximately equitable up to any vertex ($\alpha$-EQX), and shows that for any $k$ less than the number of vertices, an $\alpha$-EQX $k$-coloring is not guaranteed to exist. Given a vertex-weighted graph $G = (V, E, w)$, a coloring $(C_1, \dots, C_k)$ is said to be $\alpha$-approximately equitable up to any vertex ($\alpha$-EQX) if, for every $i, j \in [k]$ and {\em every} vertex $\overline{v} \in C_i$, we have $w(C_i \setminus \overline{v}) \le \alpha w(C_j)$. This is a stronger criterion than equitability up to one vertex. 

The following instance shows that, for any $\alpha \geq 1$, the only $\alpha$-EQX coloring puts each vertex in its own unique color class.

\begin{example}\label{eg:eqx}
For any given $\alpha \ge 1$, we provide an instance in which no $\alpha$-EQX $k$-coloring exists for any $k < n$. Consider an $n$-vertex line graph with vertices $\{1, 2, \dots, n\}$ and edges $(i, i+1)$ for each $1 \leq i < n$. Assign to each vertex $i \in [n]$ a weight $\beta^i$, where $\beta > \alpha + 1$. Note that such a choice of $\beta$ satisfies
\begin{align}
\beta^j > \alpha \left( \beta^{j-1} + \beta^{j-2} + \cdots + \beta \right), \label{eq:alpha-eqx}
\end{align}
for every positive integer $j$.

Assume, towards a contradiction, that this instance admits an $\alpha$-EQX coloring $(C_1, \dots, C_k)$ for some $k < n$. In such a coloring, at least one color class must contain at least two vertices. Let $C_h$ denote the highest weight color class that contains at least two vertices and let $i^*$ be the highest weight vertex in $C_h$. Note that $i^* \ge 3$ since $|C_h| \ge 2$ and there is an edge between vertex $1$ and vertex $2$. 

Let $C_{\ell}$ denote the color class that contains the vertex $i^* - 1$. Since an edge exists between $i^* - 1$ and $i^*$, we must have $\ell \ne h$. Additionally, since $C_{h}$ is the highest weight color class that contains multiple vertices, we must have $C_{\ell} \subseteq \{1, 2, \dots, i^* - 1\}$. Let $\overline{v}$ denote any vertex in $C_h \setminus \{i^*\}$. By the definition of $\alpha$-EQX, we have
\begin{align}
    w(C_{h} \setminus \{\overline{v}\}) \le \alpha w(C_{\ell}) \label{ineq:eqx-wts}
\end{align}
In addition, $w(C_h \setminus \{\overline{v}\}) \ge \beta^{i^*}$ and $w(C_{\ell}) \le \beta^{i^* - 1} + \beta^{i^* - 2} + \ldots + \beta^1$. These bounds, together with inequality (\ref{ineq:eqx-wts}), contradict Equation \eqref{eq:alpha-eqx}. This contradiction shows that $(C_1,\dots,C_k)$ cannot be an $\alpha$-EQX coloring. In fact, the argument implies that any $\alpha$-EQX coloring must satisfy $k = n$, with $C_i = \{i\}$ for all $i \in [n]$.
\end{example}

\end{document}

%% file: preamble.tex
\usepackage[margin=1in]{geometry}

\usepackage[hidelinks]{hyperref}
\usepackage{bbding}
\usepackage[nointegrals]{wasysym}
\usepackage[utf8]{inputenc}
\usepackage{xspace}
\usepackage{mathtools}
\usepackage{url}
\usepackage{subcaption}
\usepackage{tikz}
\usetikzlibrary{arrows, fit, positioning}
\usetikzlibrary{backgrounds,automata,calc}
\usetikzlibrary{decorations.pathreplacing}
\usepackage{csquotes}
\usepackage{amsmath,amsfonts,amsthm}
\usepackage{bm,bbm}
\usepackage{pgfplots}
\pgfplotsset{compat=1.17}

\usepackage{enumerate}
\usepackage{paralist}
\usepackage[shortlabels]{enumitem}
\usepackage{appendix}
\usepackage[capitalize]{cleveref}
\usepackage{thmtools}
\usepackage{mathtools}
\usepackage{natbib}
\usepackage{changepage}
\usepackage{authblk}
\usepackage{commath}

\usepackage{algorithm}
\usepackage{algpseudocode}

\usepackage{charter}

\newcommand{\EQo}{{\rm EQ1}\xspace}

\newcommand{\E}{\mathbb{E}}

\renewcommand{\cal}[1]{\mathcal{#1}}

\newcommand{\vare}{\varepsilon}

\newcommand{\maxw}{\mathrm{max\text{-}wt}}
\newcommand{\minw}{\mathrm{min\text{-}wt}}

\usepackage{mathtools}

\setlist{topsep=0.5ex,itemsep=0.1ex}

\newtheorem{theorem}{Theorem}[section]
\newtheorem{lemma}[theorem]{Lemma}

\newtheorem{claim}[theorem]{Claim}

\newtheorem{corr}[theorem]{Corollary}
\newtheorem*{theorem*}{Theorem}

\theoremstyle{definition}
\newtheorem{remark}[theorem]{Remark}
\newtheorem{definition}[theorem]{Definition}

\newtheorem{example}[theorem]{Example}

%% file: abb.bib
@STRING{agents = { Annual Conference on Autonomous Agents (AGENTS)}}

@STRING{ec = { ACM Conference on Economics and Computation (EC)}}

@STRING{ecai = { European Conference on Artificial Intelligence (ECAI)}}

@STRING{icalp = { International Colloquium on Automata, Languages and Programming (ICALP)}}

@STRING{ijcai = { International Joint Conference on Artificial Intelligence (IJCAI)}}

@STRING{proc = {Proceedings of the }}

@STRING{random = { International Workshop on Randomization and Computation (RANDOM)}}


%% file: references.bib
@article{HajnalSzemeredi1970,
  author  = {Hajnal, Andr{\'a}s and Szemer{\'e}di, Endre},
  title   = {Proof of a conjecture of {Erd\H{o}s}},
  journal = {Combinatorial Theory and its Applications},
  year    = {1970},
  pages   = {601--623}
}

@article{alon1987splitting,
  title={Splitting necklaces},
  author={Alon, Noga},
  journal={Advances in Mathematics},
  volume={63},
  number={3},
  pages={247--253},
  year={1987},
  publisher={Elsevier}
}

@inproceedings{Barman2023FeasiblEFX,
author = {Barman, Siddharth and Khan, Arindam and Shyam, Sudarshan and Sreenivas, K. V. N.},
title = {Guaranteeing Envy-Freeness under Generalized Assignment Constraints},
year = {2023},
booktitle = proc # {24th} # ec,
pages = {242--269}
}

@inproceedings{Pemmaraju2008SymmetryBreaking,
author = {Pemmaraju, Sriram and Srinivasan, Aravind},
title = {The Randomized Coloring Procedure with Symmetry-Breaking},
year = {2008},
booktitle = proc # {35th} # icalp,
pages = {306--319}
}

@article{Azuma1967,
  author  = {Azuma, Kazuoki},
  title   = {Weighted Sums of Certain Dependent Random Variables},
  journal = {T\^ohoku Mathematical Journal},
  volume  = {19},
  number  = {3},
  pages   = {357--367},
  year    = {1967}
}

@article{Brooks1941Colouring,
  title   = {On colouring the nodes of a network},
  author  = {Brooks, Rowland Leonard},
  journal = {Mathematical Proceedings of the Cambridge Philosophical Society},
  volume  = {37},
  number  = {2},
  pages   = {194--197},
  year    = {1941},
  publisher = {Cambridge University Press}
}

@article{Hummel2022ConflictingItems,
  author    = {Hummel, Halvard and Hetland, Magnus Lie},
  title     = {Fair allocation of conflicting items},
  journal   = {Autonomous Agents and Multi-Agent Systems},
  volume    = {36},
  number    = {8},
  year      = {2022}
}

@article{Janson2004TailBounds,
author = {Janson, Svante},
title = {Large deviations for sums of partly dependent random variables},
journal = {Random Structures \& Algorithms},
volume = {24},
number = {3},
pages = {234--248},
year = {2004}
}

@inproceedings{Pemmaraju2001TailBounds,
	author = {Pemmaraju, Sriram V.},
	booktitle = {Approximation, Randomization, and Combinatorial Optimization: Algorithms and Techniques},
	pages = {285--296},
	title = {Equitable Coloring Extends {C}hernoff-{H}oeffding Bounds},
	year = {2001}
}

@article{Hoeffding1963Inequality,
  title     = {Probability Inequalities for Sums of Bounded Random Variables},
  author    = {Hoeffding, Wassily},
  journal   = {Journal of the American Statistical Association},
  volume    = {58},
  number    = {301},
  pages     = {13--30},
  year      = {1963},
  publisher = {Taylor \& Francis}
}

@inproceedings{Chiarelli2020ConflictingItems,
author = {Chiarelli, Nina and Krnc, Matja\v{z} and Milani\v{c}, Martin and Pferschy, Ulrich and Piva\v{c}, Nevena and Schauer, Joachim},
title = {Fair Packing of Independent Sets},
year = {2020},
publisher = {Springer-Verlag},
booktitle = {Combinatorial Algorithms: 31st International Workshop (IWOCA 2020)},
pages = {154--165}
}

@article{Cechlarova2012EquitableDivisible,
	author = {Katar{\'\i}na Cechl{\'a}rov{\'a} and Eva Pill{\'a}rov{\'a}},
	journal = {Discrete Optimization},
	number = {4},
	pages = {249--257},
	title = {On the computability of equitable divisions},
	volume = {9},
	year = {2012}}

@article{Dubins1961EquitableDivisible,
 author = {L. E. Dubins and E. H. Spanier},
 journal = {The American Mathematical Monthly},
 number = {1},
 pages = {1--17},
 publisher = {[Taylor & Francis, Ltd., Mathematical Association of America]},
 title = {How to Cut A Cake Fairly},
 volume = {68},
 year = {1961}
}

@inproceedings{Freeman2019Equitable,
	author = {Freeman, Rupert and Sikdar, Sujoy and Vaish, Rohit and Xia, Lirong},
	booktitle = proc # {28th} # ijcai,
	pages = {280--286},
	title = {Equitable Allocations of Indivisible Goods},
	year = {2019}}

@inproceedings{Gourves2014Equitable,
author = {Gourv\`{e}s, Laurent and Monnot, J\'{e}r\^{o}me and Tlilane, Lydia},
title = {Near fairness in matroids},
year = {2014},
booktitle = proc # {21st} # ecai,
pages = {393--398}
}

@article{Barman2026Equitable,
author = {Barman, Siddharth and Bhaskar, Umang and Pandit, Yeshwant and Pyne, Soumyajit},
title = {Nearly Equitable Allocations Beyond Additivity and Monotonicity},
year = {2026},
volume = {85},
journal = {Journal of Artificial Intelligence Research},
numpages = {29}
}
